\documentclass[journal]{IEEEtran}

\usepackage{kotex}
\usepackage{graphicx}
\usepackage{amsmath}
\usepackage{amssymb}
\usepackage{amsthm} 
\usepackage{algorithm}
\usepackage{algpseudocode}
\usepackage{tabularx}
\usepackage[table]{xcolor}
\usepackage{hhline} 
\usepackage{subcaption}
\usepackage{amsfonts}
\usepackage{cite}
\usepackage{color}
\usepackage{url} 
\usepackage{tcolorbox}
\usepackage[dvipsnames]{xcolor}
\usepackage{booktabs}
\usepackage{multirow}
\usepackage{dsfont}

\newtheorem{theorem}{Theorem}
\newtheorem{lemma}{Lemma}
\newtheorem{proposition}{Proposition}

\algnewcommand\Input{\item[\textbf{Input:}]}
\algnewcommand\Output{\item[\textbf{Output:}]}

\algnewcommand{\IIf}[1]{\State\algorithmicif\ #1\ \algorithmicthen}
\algnewcommand{\ELIIf}[1]{\State\algorithmicelse\ \algorithmicif\ #1\ \algorithmicthen}
\algnewcommand{\ElseIIf}[1]{\algorithmicelse\ #1} 
\algnewcommand{\EndIIf}{\unskip\ \algorithmicend\ \algorithmicif}

\begin{document}
\title{Semantic-Aware Data-Aided Channel Estimation with Large Language Models for MIMO Systems}

\author{Sojeong Park,~\IEEEmembership{Graduate Student Member,~IEEE}, Jaehyun Choi, and Hyun Jong Yang,~\IEEEmembership{Senior Member,~IEEE}

\thanks{An earlier version of this paper was presented at the IEEE International Conference on Acoustics, Speech, and Signal Processing (ICASSP), 2026~\cite{park2026semantic}.

Sojeong Park is with the Department of Electrical Engineering, Pohang University of Science and Technology (POSTECH), Republic of Korea (e-mail: sojeong@postech.ac.kr).
Jaehyun Choi is with the Department of Electrical and Computer Engineering, Seoul National University, Republic of Korea (email: jhchoi0226@snu.ac.kr).
Hyun Jong Yang is with the Department of Electrical and Computer Engineering and Institute of New Media and Communications, Seoul National University, Republic of Korea (email: hjyang@snu.ac.kr).
(\textit{Corresponding author: Hyun Jong Yang})
}
}

\maketitle

\begin{abstract}
Data-aided channel estimation enhances spectral efficiency by reusing detected symbols as virtual pilots. In this process, selecting only reliable symbols is crucial to prevent misdetected symbols from corrupting the channel estimate. However, conventional methods rely exclusively on physical-layer statistics. Beyond physical-layer information, transmitted payloads possess inherent semantic structures that can be exploited to resolve detection errors. In this paper, we propose a novel semantic-aware channel estimation framework for multiple-input multiple-output (MIMO) systems that utilizes a fine-tuned large language model (LLM) to perform reliable symbol selection and correction based on semantic information. The framework employs a two-layer mechanism: one layer selects reliable decoded symbols through semantic verification, while the other selects accurately LLM-corrected symbols by cross-validating them against the received signal using physical-layer information. We prove that corrected symbols yield a strictly larger expected reduction in estimation error than initially correctly decoded symbols. Extensive simulations demonstrate that the proposed framework significantly outperforms conventional data-aided schemes in both normalized mean squared error and bit error rate, closely approaching the performance of an oracle estimator.
\end{abstract}


\section{Introduction}
\label{sec:introduction}

\IEEEPARstart{L}{arge} language models (LLMs) are now central to a wide range of AI services, including conversational assistants, coding copilots, and agentic AI applications~\cite{zhao2023survey, jiang2026survey, wang2024survey, ouyang2022training}. Since LLM inference often exceeds the capacity of edge devices, these services predominantly run on remote servers, with recent work envisioning device-edge-cloud architectures for next-generation networks~\cite{lin2025pushing, chen2024enabling}. In such distributed architectures, natural language text, such as user prompts, model responses, and inter-agent messages, must be continuously exchanged over wireless networks. Unlike conventional data, this natural language carries rich statistical and contextual structure that conveys semantic meaning beyond the bit level. Such properties align with the research direction of semantic communication, which aims to design wireless systems based on semantic meaning rather than raw bit fidelity~\cite{gunduz2022beyond, yang2023semantic, chaccour2025less}.

Various approaches have been proposed for semantic communication~\cite{xie2021deep, bourtsoulatze2019deep, weng2021semantic, jiang2022wireless, liang2026generative, qin2026generative, park2026robust,li2026llm, hao2026semantic, kim2024pilot}.
A common approach is the design of deep learning-based source encoders that compress raw inputs into compact representations preserving semantic information, with representative examples spanning text~\cite{xie2021deep}, images~\cite{bourtsoulatze2019deep}, speech~\cite{weng2021semantic}, and video~\cite{jiang2022wireless}.
More recent studies use generative models and LLMs as semantic decoders to reconstruct semantically consistent content from compressed or channel-corrupted representations, using priors learned from the source distribution ~\cite{liang2026generative, qin2026generative, park2026robust}.
Related work has also applied LLMs to channel decoding, using the contextual structure of text payloads to support error correction ~\cite{li2026llm, hao2026semantic}. These schemes either embed language model priors directly within the channel decoder~\cite{li2026llm} or apply post-decoding semantic reconstruction with channel-level verification~\cite{hao2026semantic}.
Channel estimation has also been studied within semantic communication systems. For example,~\cite{kim2024pilot} jointly designs a deep learning-based channel estimator and a reinforcement learning-based pilot allocation strategy in a semantic communication framework. However, that approach refines the channel estimate using pilot signals alone without leveraging the semantic information of the payload.
More broadly, the use of payload semantics as direct evidence within the channel estimation stage has remained largely unexplored.

Beyond semantic communication, LLMs have also been applied to a variety of wireless physical-layer tasks. These include channel prediction~\cite{liu2024llm4cp}, beam prediction~\cite{sheng2025beam}, channel state information (CSI) feedback~\cite{cui2025exploring}, multi-task physical-layer processing~\cite{zheng2026large, liu2025llm4wm}, and network-level optimization~\cite{shao2024wirelessllm, noh2025adaptive}.
One line of work repurposes an LLM as a sequence model for numerical channel data. For example,~\cite{liu2024llm4cp} fine-tunes a pretrained LLM to forecast future CSI from historical channel sequences. Along the same direction,~\cite{sheng2025beam} converts past beam indices and angles of departure into text-like tokens for mmWave beam prediction. Similarly,~\cite{cui2025exploring} treats each per-subcarrier CSI vector as a token to reconstruct the compressed CSI. Other studies extend this idea to multiple tasks at once, adapting a single LLM backbone to several channel-associated tasks within a unified framework~\cite{zheng2026large, liu2025llm4wm}.
Another line of work instead leverages the language and reasoning abilities of LLMs to assist network-level decision making. In this category, LLMs are used to handle tasks such as power allocation, resource allocation, and protocol understanding by interpreting task descriptions provided in natural language~\cite{shao2024wirelessllm, noh2025adaptive}.
Despite these advances, these approaches rely only on physical-layer information and leave the structure of the transmitted data unexploited.

Channel estimation accuracy fundamentally limits the reliability of multiple-input multiple-output (MIMO) systems, motivating extensive research beyond conventional pilot-based estimation. One direction applies deep learning to improve estimation accuracy directly from received signals. For example,~\cite{soltani2019deep} treats the time–frequency channel response as an image and applies super-resolution and denoising networks, and~\cite{hu2021deep} further analyzes such estimators, showing that a ReLU network asymptotically approaches the minimum mean-square error estimator. Data-aided channel estimation offers a complementary approach that reuses the decoded data as additional references, reducing pilot overhead~\cite{zhao2008iterative, he2020channel}. However, erroneous detected symbols admitted as virtual pilots degrade the channel estimate. Consequently, subsequent work focuses on selecting only the reliable subset~\cite{park2015iterative, weisser2024data, kim2023semi, hashempoor2025deep}.
These approaches implement reliability-based selection in different ways. For instance,~\cite{zhao2008iterative} iteratively refines the channel estimate by combining preamble, pilot, and soft-decoded data symbols in a turbo loop, and~\cite{he2020channel} feeds detected data to the estimator within a joint estimation and detection architecture. \cite{park2015iterative} treats high-reliability data symbols from an iterative detection and decoding scheme as virtual pilots, while~\cite{weisser2024data} exploits the subspace spanned by all received symbols to improve a Gaussian-mixture-model estimator. Recent work has also considered learning-based selection criteria. Reference~\cite{kim2023semi} models reliable symbol selection as a Markov decision process and applies reinforcement learning, while~\cite{hashempoor2025deep} uses a classifier trained on denoised received signals to identify reliably detected symbols.
In all such methods, however, the selection criterion is derived exclusively from physical-layer statistics. As a result, these methods can only retain the symbols that the detector already decoded correctly, with no mechanism to recover the misdetected ones.

To address these limitations, we propose a novel semantic-aware channel estimation framework that actively integrates the semantic information of the text payload utilizing an LLM. At the receiver, channel impairments manifest as typographical errors in the decoded text. These errors can be identified using the surrounding context, enabling the recovery of the corrupted characters. Based on this property, we use an LLM to verify and correct the decoded text. Consequently, this approach shifts the data-aided estimation paradigm from simply selecting reliable physical-layer detections to actively recovering misdetected symbols.

The overall procedure of the proposed framework is organized into consecutive steps. First, the base station (BS) obtains an initially decoded text using a conventional pilot-based channel estimate. To address the typographical errors within this text, a fine-tuned LLM deployed at the BS generates a corrected sequence.The LLM operates at the token level, which often results in length mismatches between the original and corrected texts. To resolve this issue, a Needleman-Wunsch (NW) sequence alignment module~\cite{needleman1970general} is applied to establish a strict character-level correspondence.
Based on this aligned text, the system constructs a \textit{semantic pilot} through a two-layer selection process. Specifically, Layer 1 identifies reliable decoded symbols by selecting characters that remain unchanged after the LLM correction. Layer 2 then selects accurately LLM-corrected symbols by validating the modifications against the received signal using physical-layer information. Ultimately, the symbols selected from both layers are utilized as additional pilots, combining with the original pilots to refine the channel estimate.

Building upon our preliminary conference publication~\cite{park2026semantic} that proposed LLM-based symbol verification for single-input single-output systems, this paper generalizes the framework to MIMO environments. Furthermore, we introduce fundamental methodological advancements not addressed in the prior work, including active LLM-driven error recovery and sequence alignment. The main contributions of this paper are summarized as follows:

\begin{itemize}
    \item \textbf{LLM-assisted framework with correction.} Unlike conventional methods that only select reliable detections, our framework exploits LLM-based text correction to actively recover misdetected symbols, introducing payload semantics as a new source of side information.

    \item \textbf{Two-layer selection mechanism.} To prevent LLM hallucinations from degrading the channel estimate, we introduce a two-layer selection mechanism. This framework integrates semantic verification in Layer 1 with physical-layer cross-validation in Layer 2, ensuring that only accurately corrected symbols are employed as additional pilots.

    \item \textbf{Sequence alignment for token-level length mismatch.} We incorporate a NW alignment module to resolve length discrepancies between the initially decoded and LLM-corrected texts. This guarantees a strict character-level correspondence necessary for accurate symbol-level selection.

    \item \textbf{Theoretical analysis of semantic pilot effectiveness.} We provide a rigorous mathematical analysis to establish the fundamental advantage of the proposed approach. Specifically, we prove that an LLM-corrected symbol yields a strictly greater expected reduction in estimation error than an initially correctly decoded symbol.
\end{itemize}

\begin{figure*}[t!]
    \centering
    \includegraphics[width=17.5cm]{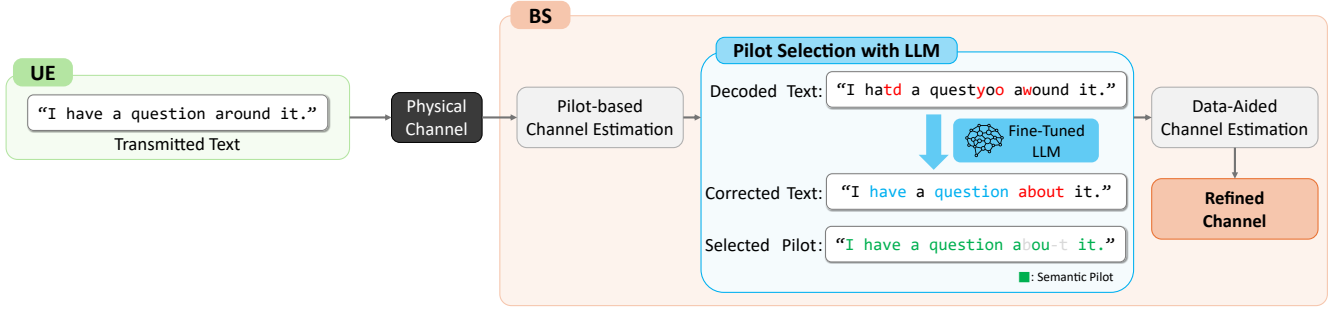}
    \caption{System model of the proposed semantic-aware data-aided channel estimation with an LLM.}
    \label{fig:system_model}
\end{figure*}

The remainder of this paper is organized as follows. Section~\ref{sec:system_model} introduces the system model and details the pipeline for semantic-aware channel estimation. Section~\ref{sec:proposed} presents the proposed framework, including the LLM-based text correction module, the NW alignment procedure, and the two-layer symbol selection mechanism. Section~\ref{sec:theoretical_analysis} provides a theoretical analysis of the semantic pilot to demonstrate the effectiveness of the corrected symbols. Section~\ref{sec:setup} details the simulation setup, Section~\ref{sec:results} reports the simulation results and ablation studies, and Section~\ref{sec:conclusion} concludes the paper.

\section{System Model}
\label{sec:system_model}

We consider an uplink MIMO system in which a user equipment (UE) equipped with $N_\mathrm{t}$ transmit antennas communicates with a BS equipped with $N_\mathrm{r}$ receive antennas. As illustrated in Fig.~\ref{fig:system_model}, the UE transmits a text payload to the BS through a physical channel. At the BS, a pilot-based initial channel estimate is obtained from a known pilot sequence and used to detect the received signal. Due to imperfect channel estimation and additive noise, the initially decoded text typically contains typographical errors. To address these errors, the BS applies an LLM to produce a corrected version of the decoded text. By leveraging the LLM-corrected text, the proposed method identifies reliable symbols in decoded sequence and selects accurately corrected symbols from the LLM output. These symbols, referred to as the semantic pilot, are used together with the original pilot to refine the initial channel estimate via data-aided channel estimation.

\subsection{Transmission Model}
The UE transmits a text payload to the BS, represented as a character sequence
\begin{equation}
    \mathbf{t} = \left[t^{(1)}, t^{(2)}, \ldots, t^{(M)}\right],
\end{equation}
where $t^{(i)}$ denotes the $i$-th character and $M$ is the length of the text. Each character $t^{(i)}$ is encoded into a bit sequence via source coding, and the resulting bits are mapped to modulation symbols drawn  from the constellation $\mathcal{X}$. The symbol stream is demultiplexed across the $N_t$ transmit antennas, forming the data symbol block $\mathbf{X}_d \in \mathcal{X}^{N_t \times L_d}$, where $L_d$ is the number of data time slots.
A pilot block $\mathbf{X}_\mathrm{p} \in \mathcal{X}^{N_\mathrm{t} \times L_\mathrm{p}}$ of length $L_\mathrm{p}$ is placed before the data block, giving the overall transmitted signal $\mathbf{X} = [\mathbf{X}_\mathrm{p}, \mathbf{X}_\mathrm{d}] \in \mathcal{X}^{N_\mathrm{t} \times (L_\mathrm{p} + L_\mathrm{d})}$.

The corresponding received signal at the BS is written as
\begin{equation}
    \mathbf{Y} = \mathbf{H}\mathbf{X} + \mathbf{N},
    \label{eq:received_signal}
\end{equation}
where $\mathbf{H} \in \mathbb{C}^{N_\mathrm{r} \times N_\mathrm{t}}$ is the MIMO channel matrix between the UE and the BS, and $\mathbf{N} \sim \mathcal{CN}(\mathbf{0}, \sigma^{2}\mathbf{I})$ is additive white Gaussian noise (AWGN). Splitting $\mathbf{Y}$ into pilot and data parts gives $\mathbf{Y} = [\mathbf{Y}_\mathrm{p}, \mathbf{Y}_\mathrm{d}]$, where $\mathbf{Y}_\mathrm{p} = \mathbf{H}\mathbf{X}_\mathrm{p} + \mathbf{N}_\mathrm{p}$ and $\mathbf{Y}_\mathrm{d} = \mathbf{H}\mathbf{X}_\mathrm{d} + \mathbf{N}_\mathrm{d}$ denote the received pilot and data blocks, respectively.

\subsection{Initial Channel Estimation and Detection}
\label{subsec:initial_estimation}
Using the pilot block $(\mathbf{X}_\mathrm{p}, \mathbf{Y}_\mathrm{p})$, the BS obtains an initial channel estimate via the linear minimum mean square error (LMMSE) estimator. The channel estimate is given by
\begin{equation}
    \hat{\mathbf{H}}_\mathrm{p}
    = \mathbf{Y}_\mathrm{p}\mathbf{X}_\mathrm{p}^{H}
    \left(\mathbf{X}_\mathrm{p}\mathbf{X}_\mathrm{p}^{H}
    + \sigma^{2}\mathbf{I}_{N_\mathrm{t}}\right)^{-1}.
    \label{eq:lmmse}
\end{equation}
The estimate $\hat{\mathbf{H}}_\mathrm{p}$ is subsequently used to detect the data symbols using maximum-a-posteriori-probability (MAP) detection. Let $\mathbf{x}_{1}, \ldots, \mathbf{x}_{K}$ denote the $K = |\mathcal{X}|^{N_\mathrm{t}}$ possible per-slot symbol vectors. Under equal priors and Gaussian noise, \begin{equation}
    \hat{\mathbf{x}}[n]
    = \arg\min_{k \in \{1,\ldots,K\}}
    \left\|\mathbf{y}[n]
    - \hat{\mathbf{H}}_\mathrm{p}\mathbf{x}_{k}\right\|^{2},
    \quad n = 1, \ldots, L_\mathrm{d},
    \label{eq:map}
\end{equation}
where $\mathbf{y}[n]$ is the $n$-th column of $\mathbf{Y}_\mathrm{d}$. The collection of all detected symbol vectors is denoted by
$\hat{\mathbf{X}}_\mathrm{d} = [\hat{\mathbf{x}}[1], \ldots, \hat{\mathbf{x}}[L_\mathrm{d}]] \in \mathcal{X}^{N_\mathrm{t} \times L_\mathrm{d}}$. The detected symbols are demodulated and source-decoded into the initially decoded text sequence
\begin{equation}
    \hat{\mathbf{t}}
    = \left[\hat{t}^{(1)}, \hat{t}^{(2)}, \ldots,
    \hat{t}^{(M)}\right].
\end{equation}
In this text sequence $\hat{\mathbf{t}}$, physical-layer errors typically appear as typographical errors.

\subsection{LLM-Based Semantic Pilot Selection}
\label{subsec:semantic_pilot}
The BS feeds the initially decoded text $\hat{\mathbf{t}}$ into an LLM deployed on the BS, which returns a corrected text sequence
\begin{equation}
    \hat{\mathbf{t}}_\mathrm{LLM}
    = \left[\hat{t}^{(1)}_\mathrm{LLM},
    \hat{t}^{(2)}_\mathrm{LLM}, \ldots,
    \hat{t}^{(M')}_\mathrm{LLM}\right],
\end{equation}
where $M'$ may differ from $M$ due to tokenization artifacts of the LLM. To resolve this length mismatch and obtain a character-level correspondence between $\hat{\mathbf{t}}$ and $\hat{\mathbf{t}}_\mathrm{LLM}$, an alignment function $\mathcal{A}$ is applied
\begin{equation}
    \bar{\mathbf{t}}, \bar{\mathbf{t}}_\mathrm{LLM}
    = \mathcal{A}\!\left(\hat{\mathbf{t}},
    \hat{\mathbf{t}}_\mathrm{LLM}\right),
\end{equation}
which returns two aligned sequences of equal length $\bar{M}$
\begin{align}
    \bar{\mathbf{t}}
    &= \left[\bar{t}^{(1)}, \bar{t}^{(2)}, \ldots,
    \bar{t}^{(\bar{M})}\right], \\
    \bar{\mathbf{t}}_\mathrm{LLM}
    &= \left[\bar{t}^{(1)}_\mathrm{LLM},
    \bar{t}^{(2)}_\mathrm{LLM}, \ldots,
    \bar{t}^{(\bar{M})}_\mathrm{LLM}\right],
\end{align}
where each $\bar{t}^{(i)}$ and $\bar{t}^{(i)}_\mathrm{LLM}$ is either a character or the gap symbol ``$-$'' inserted to compensate for insertions or deletions between $\hat{\mathbf{t}}$ and $\hat{\mathbf{t}}_\mathrm{LLM}$. For example, the decoded substring ``\texttt{awound}'' corrected by the LLM to ``\texttt{about}'' is aligned as ``\texttt{awound}'' and ``\texttt{abou-t}'', so that matched positions can be compared character by character. Based on the aligned pairs, a semantic pilot selection function $\mathcal{S}$ identifies two types of reliable symbols: reliable decoded symbols and reliable LLM-corrected symbols. The selected symbols are written as
\begin{equation}
    \mathbf{X}_\mathrm{s}
    = \mathcal{S}\!\left(\bar{\mathbf{t}},\;
    \bar{\mathbf{t}}_\mathrm{LLM},\;
    \mathbf{Y},\;
    \mathbf{X}_\mathrm{p}\right).
    \label{eq:semantic_pilot}
\end{equation}
We refer to $\mathbf{X}_\mathrm{s}$ as the semantic pilot. The detailed design of $\mathcal{A}$ and $\mathcal{S}$ is presented in Section~\ref{sec:proposed}.

\subsection{Data-Aided Channel Estimation}
\label{subsec:refinement}
By combining the semantic pilot $\mathbf{X}_\mathrm{s}$ and its corresponding received signal $\mathbf{Y}_\mathrm{s}$ with the original pilot block, we refine the channel estimate using the LMMSE estimator as follows
\begin{equation}
    \hat{\mathbf{H}}_\mathrm{LLM}
    = \left(\mathbf{Y}_\mathrm{p}\mathbf{X}_\mathrm{p}^{H}
    + \mathbf{Y}_\mathrm{s}\mathbf{X}_\mathrm{s}^{H}\right)
    \left(\mathbf{X}_\mathrm{p}\mathbf{X}_\mathrm{p}^{H}
    + \mathbf{X}_\mathrm{s}\mathbf{X}_\mathrm{s}^{H}
    + \sigma^{2}\mathbf{I}_{N_\mathrm{t}}\right)^{-1}.
    \label{eq:refined_estimation}
\end{equation}
The resulting refined channel estimate $\hat{\mathbf{H}}_\mathrm{LLM}$ serves as the final output of the proposed framework.

\section{Proposed Framework}
\label{sec:proposed}

\begin{figure}[t]
    \centering
    \includegraphics[width=0.9\linewidth]{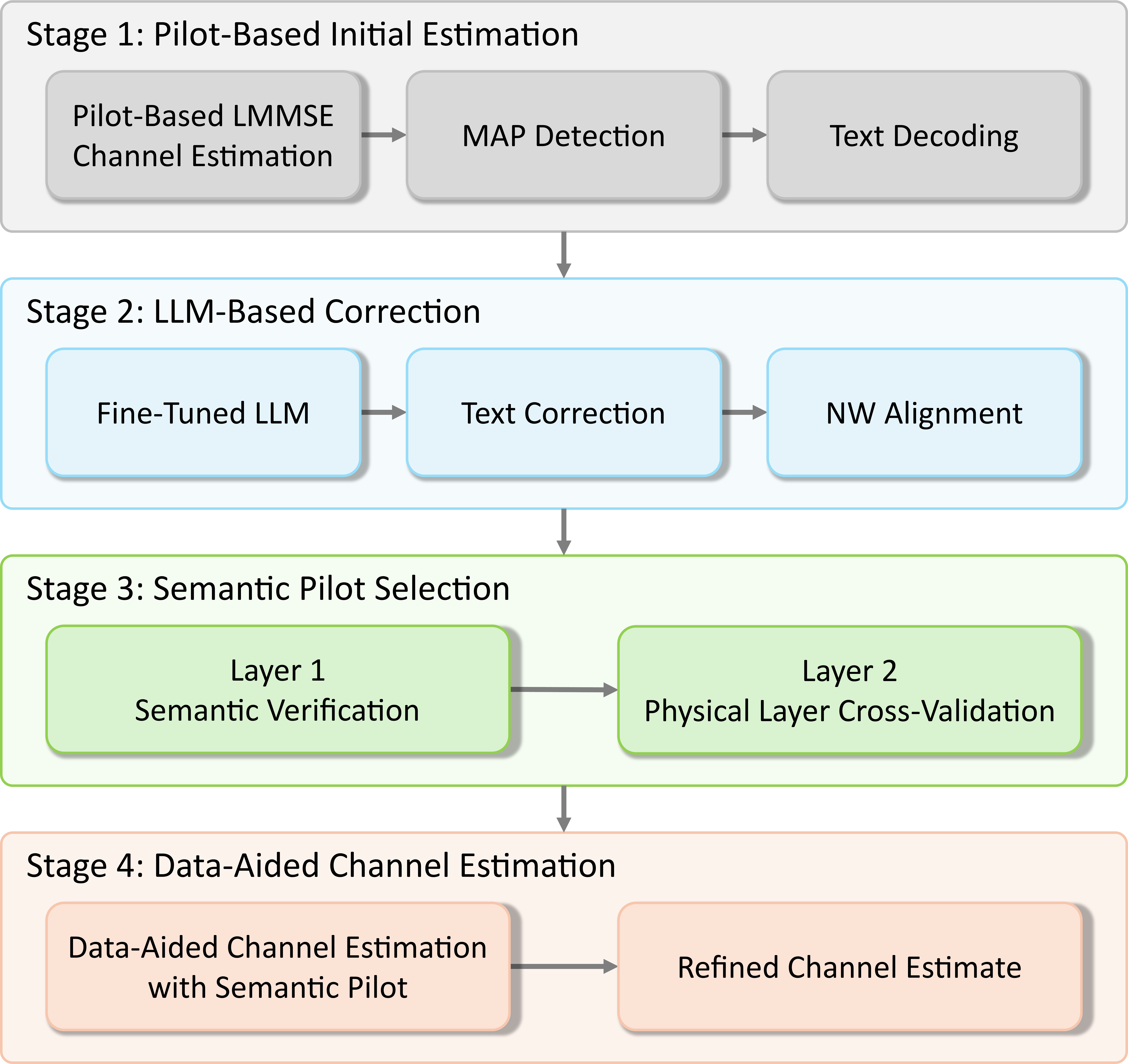}
    \caption{Overall block diagram of the proposed semantic-aware data-aided channel estimation framework.}
    \label{fig:flowchart}
\end{figure}

This section presents the detailed design of the proposed framework. The framework consists of three stages: LLM-based text correction of the initially decoded text, sequence alignment to resolve the length mismatch between the decoded and corrected texts, and semantic pilot selection. An overview of the overall pipeline is illustrated in Fig.~\ref{fig:flowchart}. The remainder of this section describes each stage in detail.

\subsection{LLM-Based Text Correction}
\label{subsec:llm_correction}

To mitigate textual corruptions in the initially decoded sequence $\hat{\mathbf{t}}$, the BS employs an LLM optimized to correct corrupted text. The LLM uses contextual dependencies to identify and correct typographical errors caused by imperfect channel estimation and noise. We adopt the open-weight Qwen3-8B model~\cite{yang2025qwen3} as the base model, adapting it to this specific error-correction task through the fine-tuning procedure described below.

\subsubsection{Prompt Design}
The LLM is provided with a task-specific prompt that instructs it to return a corrected version of the decoded text under a substitution-only constraint. As shown in Fig.~\ref{fig:prompt}, the prompt supplies the decoded text together with its length, where the placeholder \texttt{\{M\}} is replaced by the length of $\hat{\mathbf{t}}$ and \texttt{\{Decoded Text\}} by $\hat{\mathbf{t}}$.

\begin{figure}[h]
\centering
\begin{tcolorbox}[width=0.9\columnwidth, colback=gray!5, colframe=blue!50!black,
                  title=Prompt Template, fonttitle=\bfseries]
\small
\textbf{System message:}
\begin{verbatim}
You are correcting text corrupted by 
channel noise. Each character is a 
printable ASCII character. Noise may 
substitute characters with incorrect 
ones.
\end{verbatim}
\textbf{User prompt:}
\begin{verbatim}
Correct this noise-corrupted text. 
Substitution only, no insertion or 
deletion. Output must be exactly 
{M} characters.

Input: {Decoded Text}
\end{verbatim}
\end{tcolorbox}
\caption{Prompt template used for LLM-based text correction.}
\label{fig:prompt}
\end{figure}

\subsubsection{LLM Fine-Tuning}
To tailor the LLM for correcting channel-induced typographical errors, we fine-tune it on a training corpus generated from the end-to-end transmission pipeline in Section~\ref{sec:system_model}. Each training sample pairs a source text with its corresponding decoded output from the pipeline. The corrupted text is utilized as the LLM input, with the original source text serving as the ground-truth target. Generating these samples across multiple signal-to-noise ratio (SNR) levels exposes the LLM to varying error rates and ensures robust performance under various operating conditions.

Let $\mathbf{u} = [u_1, u_2, \ldots, u_{|\mathbf{u}|}]$ denote the input token sequence constructed by inserting a corrupted text into the prompt template, and let $\mathbf{v} = [v_1, v_2, \ldots, v_{|\mathbf{v}|}]$ denote the target token sequence, i.e., the tokenized uncorrupted source text. The fine-tuning loss is defined as the token-level cross-entropy applied only to the target tokens, given by
\begin{equation}
    \mathcal{L}
    = -\frac{1}{|\mathbf{v}|}
    \sum_{j=1}^{|\mathbf{v}|}
    \log p_{\boldsymbol{\theta}}\!
    \left(v_j \,\middle|\, \mathbf{u}, v_{1:j-1}\right),
    \label{eq:ft_loss}
\end{equation}
where $p_{\boldsymbol{\theta}}(\cdot)$ is the conditional distribution predicted by the LLM and $v_{1:j-1} = [v_1, \ldots, v_{j-1}]$ denotes the previously generated target tokens. The trainable parameters $\boldsymbol{\theta}$ correspond to the low-rank adapters introduced by LoRA~\cite{hu2022lora}, while the remaining pre-trained weights are frozen in quantized form following QLoRA~\cite{dettmers2023qlora}. By excluding the prompt tokens $\mathbf{u}$ from the loss calculation, we ensure the model focuses on generating the corrected text rather than repeating the prompt.

\subsection{Sequence Alignment via Needleman--Wunsch Algorithm}

This subsection describes the alignment function $\mathcal{A}$ introduced in Section~\ref{subsec:semantic_pilot}. LLMs operate on tokens rather than individual characters. Since tokens typically represent entire words, the corrections of LLMs manifest as word-based modifications. These modifications induce a length discrepancy between $\hat{\mathbf{t}}_\mathrm{LLM}$ and $\hat{\mathbf{t}}$. In the symbol domain, this mismatch introduces further errors instead of rectifying the initial errors. Therefore, aligning the two sequences is essential to maintain the corrective benefits of the LLM.

To achieve this alignment, we adopt the NW algorithm~\cite{needleman1970general}, a dynamic programming method that synchronizes the two sequences by inserting gaps to account for any added or removed characters. This process establishes a character-wise mapping that ensures corresponding characters are correctly paired despite the differences in sequence length. The algorithm identifies the optimal alignment by maximizing an alignment score, which rewards character matches while penalizing mismatches and gap insertions. Specifically, we utilize an affine gap penalty model that penalizes starting a new gap more strictly than extending a gap. This configuration encourages the algorithm to group gaps together rather than creating multiple separate gaps. This design targets word-level corrections, since replacing a single word can modify multiple consecutive characters.

The NW algorithm with affine gap penalty maintains three score matrices. For each position $(i, j)$, the match/mismatch matrix $\mathbf{D}$, the vertical gap matrix $\mathbf{P}$, and the horizontal gap matrix $\mathbf{Q}$ are updated as
\begin{align}
    D_{i,j} &= s_{i,j} + \max\!\big(
        D_{i-1,j-1},\; P_{i-1,j-1},\; Q_{i-1,j-1}
    \big), \label{eq:nw_M} \\
    P_{i,j} &= \max\!\big(
        D_{i-1,j} + g_\mathrm{o},\;
        P_{i-1,j} + g_\mathrm{e}
    \big), \label{eq:nw_X} \\
    Q_{i,j} &= \max\!\big(
        D_{i,j-1} + g_\mathrm{o},\;
        Q_{i,j-1} + g_\mathrm{e}
    \big), \label{eq:nw_Y}
\end{align}
where $s_{i,j}$ is the score for aligning the $i$-th decoded character with the $j$-th corrected character, defined as
\begin{equation}
    s_{i,j} = \begin{cases}
        s_\mathrm{match} & \text{if } \hat{t}^{(i)} =
            \hat{t}^{(j)}_\mathrm{LLM}, \\
        s_\mathrm{mis}   & \text{otherwise},
    \end{cases}
\end{equation}
with $s_\mathrm{match}$ and $s_\mathrm{mis}$ denoting the match reward and mismatch penalty, respectively. The parameters $g_\mathrm{o}$ and $g_\mathrm{e}$ denote the gap open penalty and gap extend penalty, respectively. The optimal score at position $(i, j)$ is given by
\begin{equation}
    F_{i,j} = \max\!\big(
        D_{i,j},\; P_{i,j},\; Q_{i,j}
    \big).
\end{equation}
Intuitively, $D_{i,j}$ represents a diagonal move that aligns $\hat{t}^{(i)}$ with $\hat{t}^{(j)}_\mathrm{LLM}$, $P_{i,j}$ represents a vertical move that inserts a gap in the corrected sequence, and $Q_{i,j}$ represents a horizontal move that inserts a gap in the decoded sequence. After filling the score matrices, the algorithm traces back from the bottom-right corner of the table to reconstruct the optimal alignment.

As an example, Fig.~\ref{fig:nw_algorithm} illustrates the alignment of the decoded substring ``\texttt{awound}'' and the LLM-corrected substring ``\texttt{about}''. Despite the length mismatch, the algorithm produces equal-length aligned sequences ``\texttt{awound}'' and ``\texttt{abou-t}'' by inserting a gap, maximizing the number of correctly matched character positions. This character-by-character alignment ensures that the subsequent symbol selection step processes correctly matched pairs.

\label{subsec:alignment}
\begin{figure}[t]
    \centering
    \includegraphics[width=0.75\linewidth]{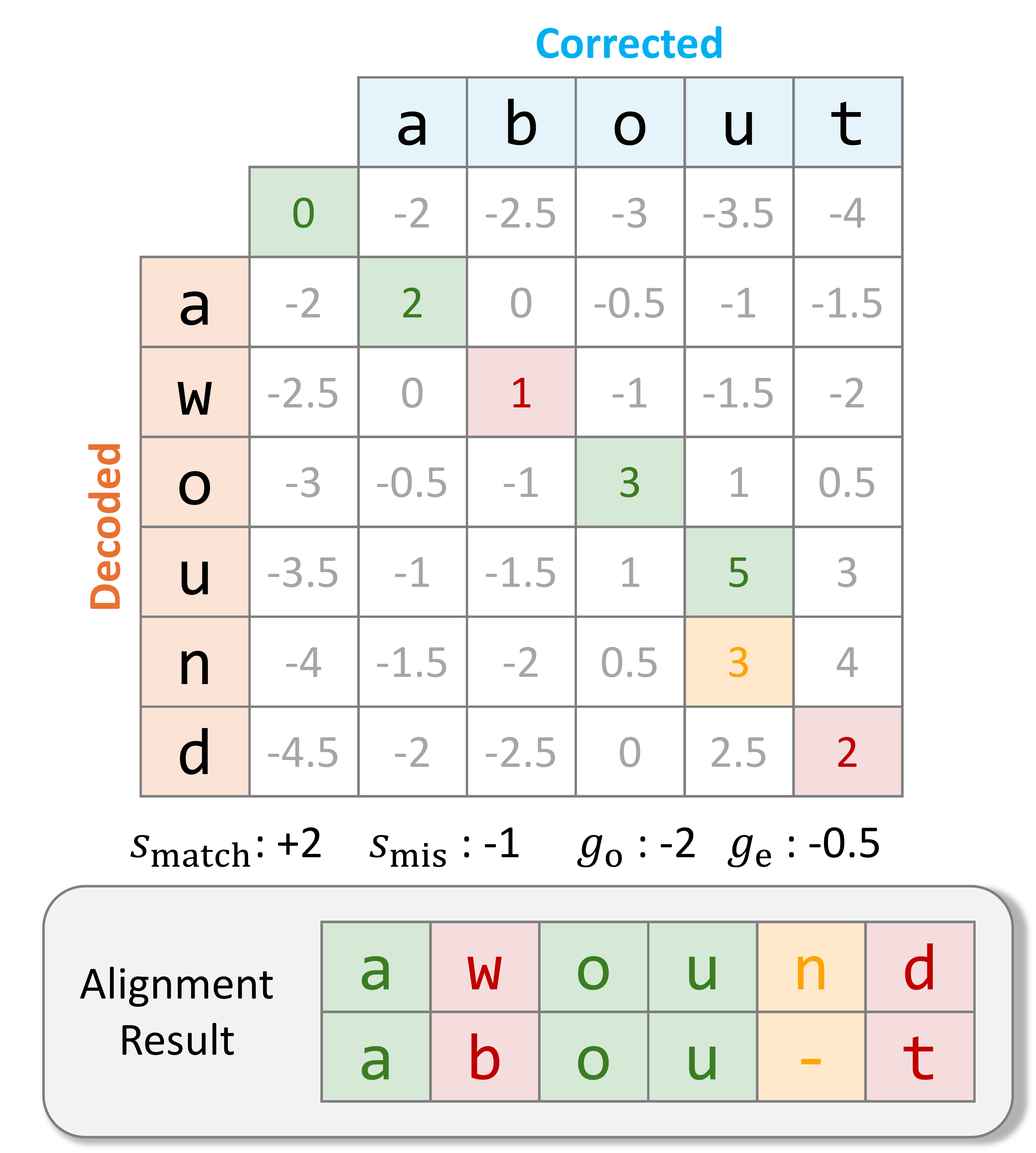}
    \caption{An example of the NW algorithm between the decoded substring `awound' and the LLM-corrected substring `about'.}
    \label{fig:nw_algorithm}
\end{figure}

\subsection{Semantic Pilot Selection}
\label{subsec:selection}

Given the aligned sequences $\bar{\mathbf{t}}$ and $\bar{\mathbf{t}}_\mathrm{LLM}$, the next step is to construct the semantic pilot by selecting reliable symbols from the data block. This subsection describes the selection function $\mathcal{S}$ defined in Section~\ref{subsec:semantic_pilot}. The selection is performed through a two-layer mechanism, as illustrated in Fig.~\ref{fig:semantic_pilot}. Layer~1 identifies reliable decoded symbols by examining the semantic agreement between the decoded and corrected texts. Layer~2 further exploits LLM corrections to obtain new reference symbols that are unavailable from physical-layer processing alone. Before incorporating LLM-based corrections into the semantic pilot, the proposed method cross-validates each correction using a physical-layer likelihood-ratio test to ensure consistency with the observed channel output. The overall procedure is summarized in Algorithm~\ref{alg:selection}, and an example of the resulting decision classes is illustrated in Fig.~\ref{fig:pilot_selection}.

\subsubsection{Layer 1 -- Semantic Verification}
For each character position $i$ in the aligned sequences $\bar{\mathbf{t}}$ and $\bar{\mathbf{t}}_\mathrm{LLM}$, the decoded and corrected characters are compared. Since the LLM is trained to correct errors, retaining a character without modification indicates that the character is consistent with its context and has been correctly decoded. Such positions are classified as
\emph{Verified} when
\begin{equation}
    \bar{t}^{(i)} = \bar{t}^{(i)}_\mathrm{LLM},
    \label{eq:layer1}
\end{equation}
where the corresponding MAP-detected symbols are subsequently included in the semantic pilot.
Let $\mathcal{I}_\mathrm{V}$ denote the set of all verified positions. The verified symbols and their corresponding received signals are collected as
\begin{align}
    \mathbf{X}_\mathrm{V}
    &= \left\{\hat{\mathbf{x}}[n]
    : n \in \mathcal{N}_i,\;
    i \in \mathcal{I}_\mathrm{V}\right\},
    \label{eq:verified_set} \\
    \mathbf{Y}_\mathrm{V}
    &= \left\{\mathbf{y}[n]
    : n \in \mathcal{N}_i,\;
    i \in \mathcal{I}_\mathrm{V}\right\},
    \label{eq:verified_recv}
\end{align}
where $\mathcal{N}_i$ denotes the set of symbol slots associated with position $i$. These verified symbols, together with the original pilot, are used to compute an intermediate refined channel estimate via the LMMSE estimator as
\begin{equation}
    \hat{\mathbf{H}}_\mathrm{V}
    = \left(\mathbf{Y}_\mathrm{p}\mathbf{X}_\mathrm{p}^{H}
    + \mathbf{Y}_\mathrm{V}\mathbf{X}_\mathrm{V}^{H}\right)
    \left(\mathbf{X}_\mathrm{p}\mathbf{X}_\mathrm{p}^{H}
    + \mathbf{X}_\mathrm{V}\mathbf{X}_\mathrm{V}^{H}
    + \sigma^{2}\mathbf{I}_{N_\mathrm{t}}\right)^{-1}.
    \label{eq:H_V}
\end{equation}
This intermediate estimate is used for physical-layer cross-validation.

\begin{figure}[t]
    \centering
    \includegraphics[width=0.9\linewidth]{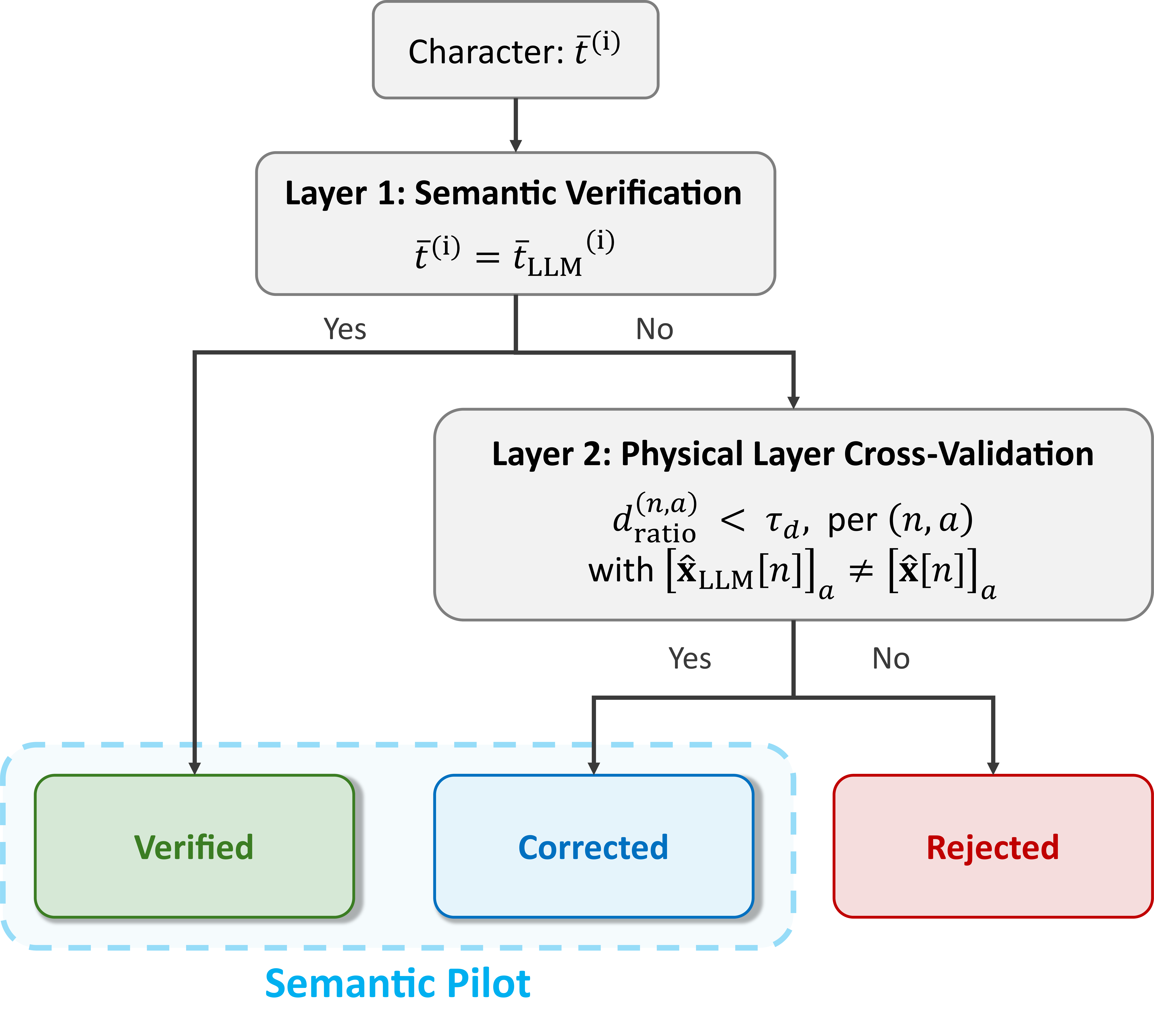}
    \caption{Flowchart of the proposed two-layer semantic pilot selection mechanism.}
    \label{fig:semantic_pilot}
\end{figure}

\begin{algorithm}[t]
\caption{Semantic Pilot Selection}
\label{alg:selection}
\begin{algorithmic}[1]
\Input $\bar{\mathbf{t}}$, $\bar{\mathbf{t}}_\mathrm{LLM}$,
$\mathbf{X}_\mathrm{p}$, $\mathbf{Y}$
\Statex \textbf{Layer 1: Semantic Verification}
\State $\mathcal{I}_\mathrm{V} \gets \{i :
\bar{t}^{(i)} = \bar{t}^{(i)}_\mathrm{LLM}\}$
\State Collect $(\mathbf{X}_\mathrm{V},
\mathbf{Y}_\mathrm{V})$ at positions
$\mathcal{I}_\mathrm{V}$ using
$\hat{\mathbf{X}}_\mathrm{d}$
\State $\hat{\mathbf{H}}_\mathrm{V} \gets \mathrm{LMMSE}
([\mathbf{X}_\mathrm{p}, \mathbf{X}_\mathrm{V}],\;
[\mathbf{Y}_\mathrm{p}, \mathbf{Y}_\mathrm{V}])$
\Statex \textbf{Layer 2: Physical-Layer Cross-Validation}
\State $\mathcal{S}_\mathrm{C} \gets \emptyset$
\For{$i$ with $\bar{t}^{(i)} \neq \bar{t}^{(i)}_\mathrm{LLM}$}
    \For{$n \in \mathcal{N}_i$ and $a \in \{1,\dots,N_t\}$
    with $[\hat{\mathbf{x}}_\mathrm{LLM}[n]]_a \neq [\hat{\mathbf{x}}[n]]_a$}
        \State Compute $d_\mathrm{ratio}^{(n,a)}$ via
        \eqref{eq:d_ratio}
        \If{$d_\mathrm{ratio}^{(n,a)} < \tau_d$}
            \State $\mathcal{S}_\mathrm{C} \gets
            \mathcal{S}_\mathrm{C} \cup \{(n, a)\}$
        \EndIf
    \EndFor
\EndFor
\State $\mathbf{x}_\mathrm{C}[n] \gets$ \eqref{eq:composite},
$\;\forall n$ with $\exists\, a \text{ s.t. } (n, a) \in \mathcal{S}_\mathrm{C}$
\State $\mathbf{X}_\mathrm{C} \gets \{\mathbf{x}_\mathrm{C}[n]\}$,
$\;\mathbf{Y}_\mathrm{C} \gets \{\mathbf{y}[n]\}$
\Output $\mathbf{X}_\mathrm{s} \gets [\mathbf{X}_\mathrm{V},
\mathbf{X}_\mathrm{C}]$
\end{algorithmic}
\end{algorithm}

\subsubsection{Layer 2 -- Physical-Layer Cross-Validation}

Although LLM modifications can recover transmitted symbols where the MAP detector fails, inaccurate corrections degrade the channel estimation accuracy. Consequently, the system must carefully filter these candidates before use. By cross-validating each corrected symbol against the received signal, the framework ensures reliability. Only the symbols that pass the cross-validation are classified as \emph{Corrected} and included in the semantic pilot.

For a symbol slot $n$ associated with position $i$, the vector $\hat{\mathbf{x}}_\mathrm{LLM}[n]$ is obtained by re-modulating the LLM-corrected character $\bar{t}^{(i)}_\mathrm{LLM}$. To evaluate this correction, a candidate vector $\tilde{\mathbf{x}}^{(n,a)}$ is constructed by replacing the $a$-th entry of the MAP-detected vector $\hat{\mathbf{x}}[n]$ with the corresponding entry from $\hat{\mathbf{x}}_\mathrm{LLM}[n]$. The cross-validation is applied only to entries where the two vectors differ, i.e., $[\hat{\mathbf{x}}_\mathrm{LLM}[n]]_a \neq
[\hat{\mathbf{x}}[n]]_a$. Unchanged entries are excluded from this process, as they provide no additional information for channel refinement. For each such entry, we define the residual ratio as
\begin{equation}
    d_\mathrm{ratio}^{(n,a)}
    = \frac{\left\|\mathbf{y}[n] - \hat{\mathbf{H}}_\mathrm{V}\,
    \tilde{\mathbf{x}}^{(n,a)}\right\|^{2}}
    {\left\|\mathbf{y}[n] - \hat{\mathbf{H}}_\mathrm{V}\,
    \hat{\mathbf{x}}[n]\right\|^{2}}.
    \label{eq:d_ratio}
\end{equation}
The residual ratio $d_\mathrm{ratio}^{(n,a)}$ quantifies the relative fit of the LLM-corrected symbol to the observed signal $\mathbf{y}[n]$ compared to the MAP-detected baseline. This metric serves as a criterion for assessing the validity of each correction, providing a basis for selecting reliable symbols for the semantic pilot set.

Each symbol-antenna pair is independently classified as
\begin{equation}
    (n, a) \in
    \begin{cases}
        \mathcal{S}_\mathrm{C}
        & \text{if } [\hat{\mathbf{x}}_\mathrm{LLM}[n]]_a \neq [\hat{\mathbf{x}}[n]]_a
        \text{ and } d_\mathrm{ratio}^{(n,a)} < \tau_d, \\
        \mathcal{S}_\mathrm{R}
        & \text{otherwise},
    \end{cases}
    \label{eq:d_ratio_test}
\end{equation}
where $\mathcal{S}_\mathrm{C}$ and $\mathcal{S}_\mathrm{R}$ denote the sets of accepted and rejected symbol-antenna pairs, respectively.
For each symbol slot $n$ that contains at least one accepted entry, a reference vector $\mathbf{x}_\mathrm{C}[n]$ is constructed by combining the accepted LLM-corrected symbols and the MAP-detected symbols within the same slot. The $a$-th entry of $\mathbf{x}_\mathrm{C}[n]$ is given by
\begin{equation}
    [\mathbf{x}_\mathrm{C}[n]]_a =
    \begin{cases}
        [\hat{\mathbf{x}}_\mathrm{LLM}[n]]_a
        & \text{if } (n, a) \in \mathcal{S}_\mathrm{C}, \\
        [\hat{\mathbf{x}}[n]]_a
        & \text{otherwise}.
    \end{cases}
    \label{eq:composite}
\end{equation}
The corrected symbols and their corresponding received signals are then collected as
\begin{align}
    \mathbf{X}_\mathrm{C}
    &= \left\{\mathbf{x}_\mathrm{C}[n]
    : \exists\, a \text{ s.t. } (n, a) \in
    \mathcal{S}_\mathrm{C}\right\},
    \label{eq:corrected_set} \\
    \mathbf{Y}_\mathrm{C}
    &= \left\{\mathbf{y}[n]
    : \exists\, a \text{ s.t. } (n, a) \in
    \mathcal{S}_\mathrm{C}\right\}.
    \label{eq:corrected_recv}
\end{align}
The final semantic pilot is then constructed as $\mathbf{X}_\mathrm{s} = [\mathbf{X}_\mathrm{V}, \mathbf{X}_\mathrm{C}]$ with the corresponding received signal $\mathbf{Y}_\mathrm{s} = [\mathbf{Y}_\mathrm{V}, \mathbf{Y}_\mathrm{C}]$. The refined channel estimate $\hat{\mathbf{H}}_\mathrm{LLM}$ is then obtained from \eqref{eq:refined_estimation} using $\mathbf{X}_\mathrm{s}$ and $\mathbf{Y}_\mathrm{s}$.

\begin{figure*}[t]
    \centering
    \includegraphics[width=17.5cm]{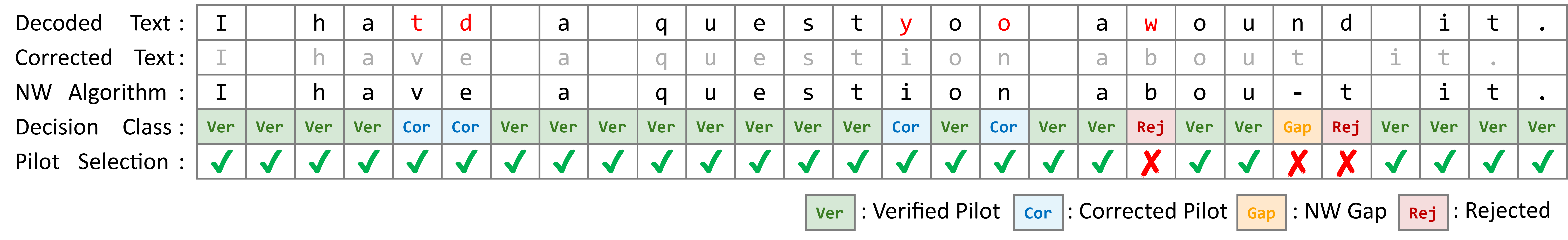}
    \caption{An example of the semantic pilot construction, demonstrating sequence alignment and the resulting decision classes for a sample sentence.}
    \label{fig:pilot_selection}
\end{figure*}

\section{Theoretical Analysis of the Semantic Pilot}
\label{sec:theoretical_analysis}
 
In this section, we show that a corrected symbol, initially misdecoded, contributes a strictly larger expected reduction in channel estimation error than a symbol that the MAP detector decodes correctly. This analysis demonstrates the effectiveness of the corrected symbols in Layer~2.

\subsection{Setup}
\label{subsec:ta_setup}

Let $\hat{\mathbf{H}}_\mathrm{p}$ denote the pilot-only estimate and $\mathbf{E} = \hat{\mathbf{H}}_\mathrm{p} - \mathbf{H}$ its error. Under the orthogonal pilot $\mathbf{X}_\mathrm{p}\mathbf{X}_\mathrm{p}^H = L_\mathrm{p}\mathbf{I}_{N_\mathrm{t}}$, substituting $\mathbf{Y}_\mathrm{p} = \mathbf{H}\mathbf{X}_\mathrm{p} + \mathbf{N}_\mathrm{p}$ into~\eqref{eq:lmmse} yields
\begin{equation}
    \mathbf{E} = -\frac{\sigma^2}{L_\mathrm{p}+\sigma^2}\mathbf{H}
    + \frac{1}{L_\mathrm{p}+\sigma^2}\mathbf{N}_\mathrm{p}\mathbf{X}_\mathrm{p}^H.
    \label{eq:ta_E}
\end{equation}
Since $\mathbf{H}$ and $\mathbf{N}_\mathrm{p}$ are independent, the rows of $\mathbf{E}$ are i.i.d. zero-mean Gaussian with per-row covariance
\begin{equation}
    \mathbf{C}_\mathrm{p} = \alpha\,\mathbf{I}_{N_\mathrm{t}},
    \label{eq:ta_Cp}
\end{equation}
where $\alpha = \sigma^2/(L_\mathrm{p} + \sigma^2)$. A direct consequence of~\eqref{eq:ta_E}, established in Appendix~\ref{app:prop_proof}, is that the pilot-only estimate $\hat{\mathbf{H}}_\mathrm{p}$ and its error $\mathbf{E}$ are statistically independent.

To measure how the pilot-only channel error interacts with a given data $\mathbf{x}$, we define the \emph{realized error alignment} as
\begin{equation}
R(\mathbf{x};\mathbf{E}) = \|\mathbf{E}\mathbf{x}\|^2.
\label{eq:ta_R}
\end{equation}
The vector $\mathbf{E}\mathbf{x}$ represents the discrepancy between the noise-free received signal $\mathbf{H}\mathbf{x}$ and its estimate $\hat{\mathbf{H}}_\mathrm{p}\mathbf{x}$ derived from the pilot-only channel. Consequently, $R(\mathbf{x};\mathbf{E})$ quantifies the signal distortion induced by the estimation error along the direction of $\mathbf{x}$.

We refer to a correctly decoded symbol as $\mathbf{x}_\mathrm{D}$ and to an initially misdecoded but corrected symbol as $\mathbf{x}_\mathrm{C}$, and assume both reference types equal the true transmitted symbol. $\Delta(\mathbf{x})$ denotes the expected reduction in $\|\hat{\mathbf{H}}_\mathrm{new}-\mathbf{H}\|_F^2$ when $\mathbf{x}$ is used as an additional pilot, where $\hat{\mathbf{H}}_\mathrm{new}$ is the resulting refined channel estimate.

\subsection{Per-Symbol Error Reduction}
\label{subsec:ta_lemma}

\begin{lemma}
\label{lem:ta_delta}
For a reference $\mathbf{x}$ equal to the true transmitted symbol with $\|\mathbf{x}\|^2
= N_\mathrm{t}$,
\begin{equation}
\Delta(\mathbf{x}) = a\,\mathbb{E}\!\left[R(\mathbf{x};\mathbf{E})
\,\big|\,\mathbf{x}\right] - b,
\label{eq:ta_delta}
\end{equation}
where
\begin{equation}
a = \frac{2(L_\mathrm{p}+\sigma^2) + N_\mathrm{t}}
        {(L_\mathrm{p}+\sigma^2 + N_\mathrm{t})^2},
\quad
b = \frac{\sigma^2\,N_\mathrm{r}\,N_\mathrm{t}}
        {(L_\mathrm{p}+\sigma^2 + N_\mathrm{t})^2}.
\label{eq:ta_ab}
\end{equation}
In particular, $\Delta(\mathbf{x})$ is a strictly increasing affine function of $\mathbb{E}[R(\mathbf{x};\mathbf{E})\mid\mathbf{x}]$.
\end{lemma}
\begin{IEEEproof}
See Appendix~\ref{app:lemma_proof}.
\end{IEEEproof}

Lemma~\ref{lem:ta_delta} quantifies the improvement of the refined estimate relative to the pilot-only estimate, showing that this improvement grows monotonically with $R(\mathbf{x};\mathbf{E})$. Intuitively, an additional reference corrects the estimation error only along its own direction. Consequently, a larger error along that direction yields a greater achievable reduction. Since the coefficients $(a,b)$ depend only on $\|\mathbf{x}\|^2 = N_\mathrm{t}$, they are common to both reference types, so the entire difference between $\mathbf{x}_\mathrm{D}$ and $\mathbf{x}_\mathrm{C}$ is determined by their respective values of $\mathbb{E}[R(\mathbf{x};\mathbf{E})]$.

\subsection{Selection Bias of the Initial Detector}
\label{subsec:ta_prop}

We next show that the two reference types differ precisely in $\mathbb{E}[R(\mathbf{x};\mathbf{E})]$. 

\begin{proposition}
\label{prop:ta_bias}
Let $\mathbf{x}_\mathrm{D}$ be a correctly decoded reference and $\mathbf{x}_\mathrm{C}$ a corrected reference. The error alignment of a corrected reference is, in expectation, strictly larger than that of a correctly decoded one,
\begin{equation}
\mathbb{E}\big[R(\mathbf{x}_\mathrm{D};\mathbf{E})\big]
< \mathbb{E}\big[R(\mathbf{x}_\mathrm{C};\mathbf{E})\big].
\label{eq:ta_bias}
\end{equation}
\end{proposition}
\begin{IEEEproof}
See Appendix~\ref{app:prop_proof}.
\end{IEEEproof}

The advantage in Proposition~\ref{prop:ta_bias} stems from a selection effect, not from any intrinsic difference between the symbols. Before any detection, $R(\mathbf{x};\mathbf{E})$ follows the same distribution for every transmitted symbol. Therefore, the difference arises solely from how each reference is selected. When the error alignment $R(\mathbf{x};\mathbf{E})$ is large, the channel estimate is heavily distorted along the direction of $\mathbf{x}$, pushing the detector toward a wrong decision. A symbol thus becomes less likely to be decoded correctly as its $R(\mathbf{x};\mathbf{E})$ grows. As a result, the correctly decoded references $\mathbf{x}_\mathrm{D}$ generally correspond to small error alignments, whereas the initially misdecoded but subsequently corrected references $\mathbf{x}_\mathrm{C}$ arise from large error alignments. The two conditional expectations therefore separate as in~\eqref{eq:ta_bias}, which Appendix~\ref{app:prop_proof} formalizes through a correlation inequality.

\subsection{Main Result}
\label{subsec:ta_main}

\begin{theorem}
\label{thm:ta_main}
The expected NMSE reduction contributed by a single corrected reference exceeds that contributed by a single correctly decoded reference, i.e.,
\begin{equation}
\mathbb{E}[\Delta(\mathbf{x}_\mathrm{C})]
> \mathbb{E}[\Delta(\mathbf{x}_\mathrm{D})].
\label{eq:ta_main_result}
\end{equation}
\end{theorem}
\begin{IEEEproof}
By Lemma~\ref{lem:ta_delta}, $\Delta(\mathbf{x})$ is strictly increasing in $\mathbb{E}[R(\mathbf{x};\mathbf{E})\mid\mathbf{x}]$, and by Proposition~\ref{prop:ta_bias}, $\mathbb{E}[R(\mathbf{x}_\mathrm{D};\mathbf{E})] < \mathbb{E}[R(\mathbf{x}_\mathrm{C};\mathbf{E})]$. Combining the two yields the result.
\end{IEEEproof}

Theorem~\ref{thm:ta_main} formalizes the fundamental advantage of the semantic pilot. Although both reference types contribute to channel refinement, their impacts differ significantly. A correctly decoded symbol aligns with directions where the initial estimation error is already minimal, resulting in only incremental gains. In contrast, a corrected symbol aligns with directions of substantial estimation error, yielding a more pronounced improvement in estimation accuracy. The proposed framework further improves channel estimation by using corrected symbols.

\section{Simulation Setup}
\label{sec:setup}
We consider a $2 \times 2$ MIMO system with QPSK modulation, resulting in $K = 16$ candidate symbol vectors per time slot. The channel $\mathbf{H} \in \mathbb{C}^{N_\mathrm{r} \times N_\mathrm{t}}$ is modeled as a frequency-flat Rayleigh fading channel with i.i.d. entries drawn from $\mathcal{CN}(0, 1)$, constant within each transmission frame. The pilot sequence is a Zadoff-Chu sequence of length $L_\mathrm{p} = 16$. Each character is encoded using 8-bit ASCII. The text corpus is drawn from the Europarl dataset~\cite{koehn-2005-europarl}, with non-overlapping splits for LLM fine-tuning, validation, and testing.

\subsection{LLM Configuration}
The base model is Qwen3-8B~\cite{yang2025qwen3}, an 8-billion-parameter open-weight LLM. The model is fine-tuned via QLoRA~\cite{dettmers2023qlora} with 4-bit NormalFloat (NF4) quantization, reducing the model size from 16.4~GB to approximately 4.1~GB while keeping only 0.5\% of the total parameters trainable through low-rank adapters. The fine-tuning configuration is summarized in Table~\ref{tab:llm_config}, and the hardware specifications are listed in Table~\ref{tab:hardware}.

\begin{table}[h]
\centering
\caption{LLM fine-tuning configuration.}
\label{tab:llm_config}
\begin{tabular}{ll}
\toprule
Parameter & Value \\
\midrule
Base model & Qwen3-8B \\
Quantization & NF4 \\
LoRA rank & 16 \\
LoRA scaling & 32 \\
Trainable parameters & $\sim$42M (0.5\%) \\
Training samples & 20,000 \\
Validation samples & 2,000 \\
Epochs & 3 \\
Effective batch size & 32 \\
Learning rate & $2 \times 10^{-4}$ \\
LR scheduler & Cosine with warmup \\
Decoding & Greedy \\
\bottomrule
\end{tabular}
\end{table}

\begin{table}[h]
\centering
\caption{Hardware specifications.}
\label{tab:hardware}
\begin{tabular}{ll}
\toprule
Item & Specification \\
\midrule
GPU & NVIDIA GeForce RTX 3090 (24~GB) \\
CPU & AMD Ryzen 7 5800X (8-core) \\
System RAM & 128~GB DDR4 \\
\bottomrule
\end{tabular}
\end{table}

\subsection{Framework Parameters}
The NW alignment scoring parameters are set to $s_\mathrm{match} = +2$, $s_\mathrm{mis} = -1$, $g_\mathrm{o} = -2$, and $g_\mathrm{e} = -0.5$. The $d$-ratio threshold for Layer~2 cross-validation is set to $\tau_d = 2$.

\subsection{Compared Methods}
The following methods are compared in the simulations.

\begin{itemize}
    \item \textbf{Pilot Only}: This baseline uses the initial channel estimate $\hat{\mathbf{H}}_\mathrm{p}$ from \eqref{eq:lmmse} directly for MAP detection without any data-aided refinement.

    \item \textbf{Semi-DA RL~\cite{kim2023semi}}: This scheme utilizes a reinforcement learning-based data-aided channel estimator. To ensure a fair comparison, the candidate pool for symbol selection includes the entire data block.

    \item \textbf{LLM All Corrected}: All symbols corrected by the LLM serve as reference symbols for data-aided estimation without applying any alignment or verification.

    \item \textbf{Layer~1 w/o NW~\cite{park2026semantic}}: This refers to our preliminary work, which applies only Layer~1 semantic verification without the NW alignment module or Layer~2 physical-layer cross-validation.

    \item \textbf{Layer~1 w/ NW}: This baseline combines Layer~1 semantic verification with the NW alignment module while excluding Layer~2 cross-validation.

    \item \textbf{MAP Oracle}: This benchmark uses only the correctly decoded MAP symbols by assuming perfect knowledge of the decoding results.
    
    \item \textbf{Perfect CSI}: The actual channel matrix $\mathbf{H}$ is used for detection.

    \item \textbf{Proposed}: This represents the complete proposed framework, integrating NW alignment, Layer~1 semantic verification, and Layer~2 physical-layer cross-validation using a fine-tuned LLM.
\end{itemize}

\begin{table*}[t]
\centering
\caption{BER comparison of compared methods. Best result in bold, second-best underlined.}
\label{tab:ber_main}
\begin{tabular}{l cccccccccc}
\toprule
Method & 1 dB & 2 dB & 3 dB & 4 dB & 5 dB
       & 6 dB & 7 dB & 8 dB & 9 dB & 10 dB \\
\midrule
Pilot Only
    & .1412 & .1170 & .0950 & .0751 & .0571
    & .0430 & .0318 & .0224 & .0151 & .0106 \\
Semi-DA RL
    & .1403 & .1145 & .0911 & .0711 & .0533
    & .0395 & .0285 & .0198 & .0135 & .0093 \\
Layer~1 w/o NW
    & .1406 & .1162 & .0939 & .0736 & .0554
    & .0414 & .0301 & .0212 & .0142 & .0098 \\
Layer~1 w/ NW
    & .1359 & .1107 & .0881 & .0688 & .0517
    & .0381 & .0278 & .0192 & .0131 & .0091 \\
LLM All Corrected
    & .1376 & .1139 & .0934 & .0740 & .0559
    & .0421 & .0327 & .0227 & .0170 & .0111 \\
Proposed
    & \underline{.1326} & \underline{.1078}
    & \underline{.0862} & \underline{.0670}
    & \underline{.0505} & \underline{.0370}
    & \textbf{.0270} & \underline{.0190}
    & \textbf{.0127} & \textbf{.0089} \\
\midrule
MAP Oracle
    & \textbf{.1324} & \textbf{.1074}
    & \textbf{.0855} & \textbf{.0671}
    & \textbf{.0504} & \textbf{.0373}
    & \underline{.0273} & \textbf{.0190}
    & \underline{.0127} & \underline{.0090} \\
Perfect CSI
    & .1260 & .1035 & .0825 & .0652 & .0489
    & .0360 & .0269 & .0186 & .0124 & .0089 \\
\bottomrule
\end{tabular}
\end{table*}

\subsection{Evaluation Metrics}
We evaluate the proposed framework with two sets of metrics.

\subsubsection{Channel estimation and detection performance}
\begin{itemize}
    \item \textbf{NMSE}: Normalized mean square error of the channel estimate, defined as $\|\hat{\mathbf{H}} - \mathbf{H}\|_F^2 /\|\mathbf{H}\|_F^2$.

    \item \textbf{BER}: Bit error rate after MAP detection with the respective channel estimate.
\end{itemize}

\subsubsection{Semantic pilot selection quality}
\begin{itemize}
    \item \textbf{Selection Ratio}: The proportion of data symbols included in the semantic pilot out of all transmitted data symbols.

    \item \textbf{Precision}: The proportion of selected semantic pilot symbols that match the true transmitted symbols.

    \item \textbf{Detection Rate}: The proportion of correctly decoded data symbols that are included in the semantic pilot, measuring the completeness of the selection.

    \item \textbf{Layer 2 Contribution}: The proportion of semantic pilot symbols that originate from Layer~2 corrections.
\end{itemize}

\section{Simulation Results}
\label{sec:results}

\subsection{Channel Estimation and Detection Performance}
\label{subsec:ce_performance}

Fig.~\ref{fig:main_nmse} compares the NMSE of the evaluated methods over SNRs of 1–10~dB. The proposed method consistently yields the lowest NMSE, notably surpassing the MAP Oracle, which serves as the ideal benchmark using all correctly decoded symbols as references. This result indicates that the LLM-aided corrections in Layer~2 provide reliable reference information that is unattainable through conventional physical-layer processing. The factors driving this improvement are analyzed in detail in Section~\ref{subsec:pilot_analysis}.
Among the baseline methods, Layer~1~w/~NW outperforms Layer~1 w/o~NW across all tested SNR levels. This confirms that the NW alignment module enhances the reliability of character-level comparison. Additionally, Layer~1~w/o~NW consistently outperforms Semi-DA RL. This indicates that semantic-based symbol selection is more effective than physical-layer-based selection in the considered simulation setup. In contrast, the LLM All Corrected method exhibits the worst performance, consistently falling below the Pilot-only baseline. This degradation is attributed to the inclusion of incorrect corrections and LLM hallucinations when all LLM outputs are utilized without a filtering mechanism. These unreliable symbols contaminate the channel estimate, demonstrating the importance of verifying corrections before incorporating them into the pilot.

Table~\ref{tab:ber_main} summarizes the BER results following MAP detection for each channel estimation scheme. The proposed framework outperforms all practical methods across all SNR levels, closely approaching the MAP Oracle and Perfect CSI bounds at high SNRs. These findings confirm that the channel estimation gains observed in the NMSE analysis directly translate into improved detection performance.

\begin{figure}[t]
    \centering
    \includegraphics[width=0.85\linewidth]{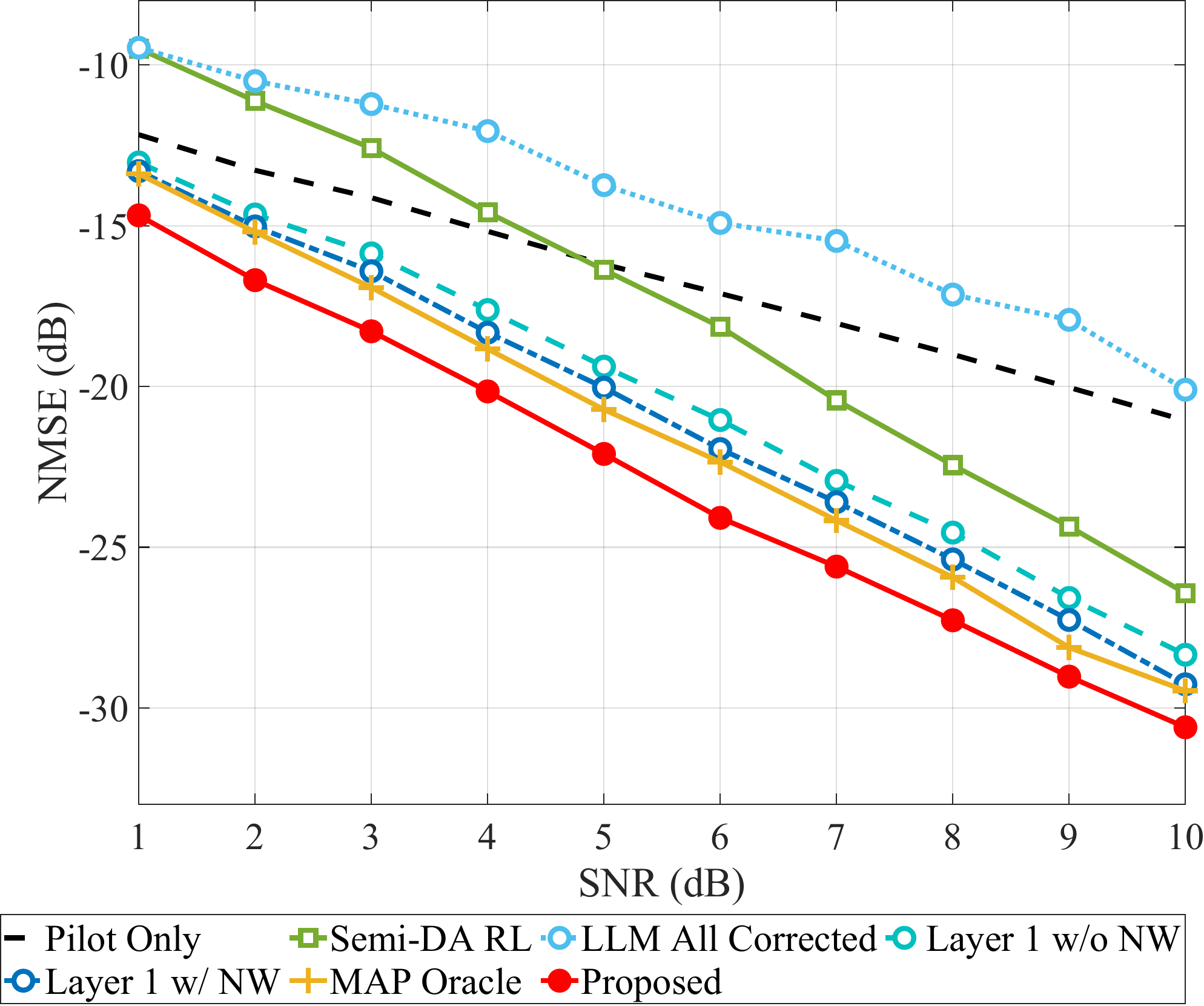}
    \caption{NMSE performance versus SNR for the compared methods.}
    \label{fig:main_nmse}
\end{figure}

\subsection{Semantic Pilot Selection Analysis}
\label{subsec:pilot_analysis}
\begin{figure}[t]
    \centering
    \begin{subfigure}{0.85\linewidth}
        \centering
        \includegraphics[width=\linewidth]{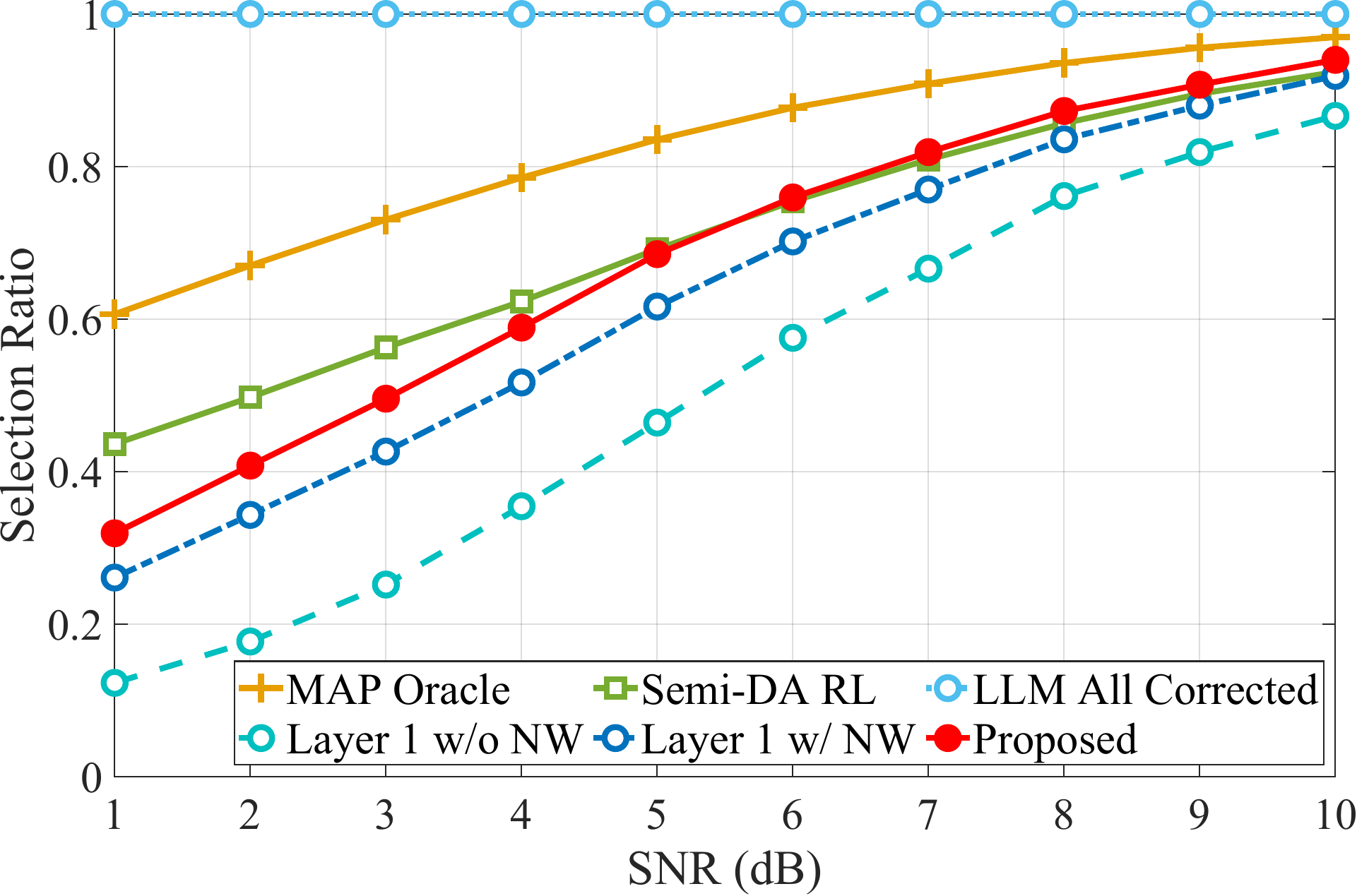}
        \caption{Selection Ratio}
        \label{fig:main_selection_ratio}
    \end{subfigure}
    
    \vspace{1mm}
    
    \begin{subfigure}{0.85\linewidth}
        \centering
        \includegraphics[width=\linewidth]{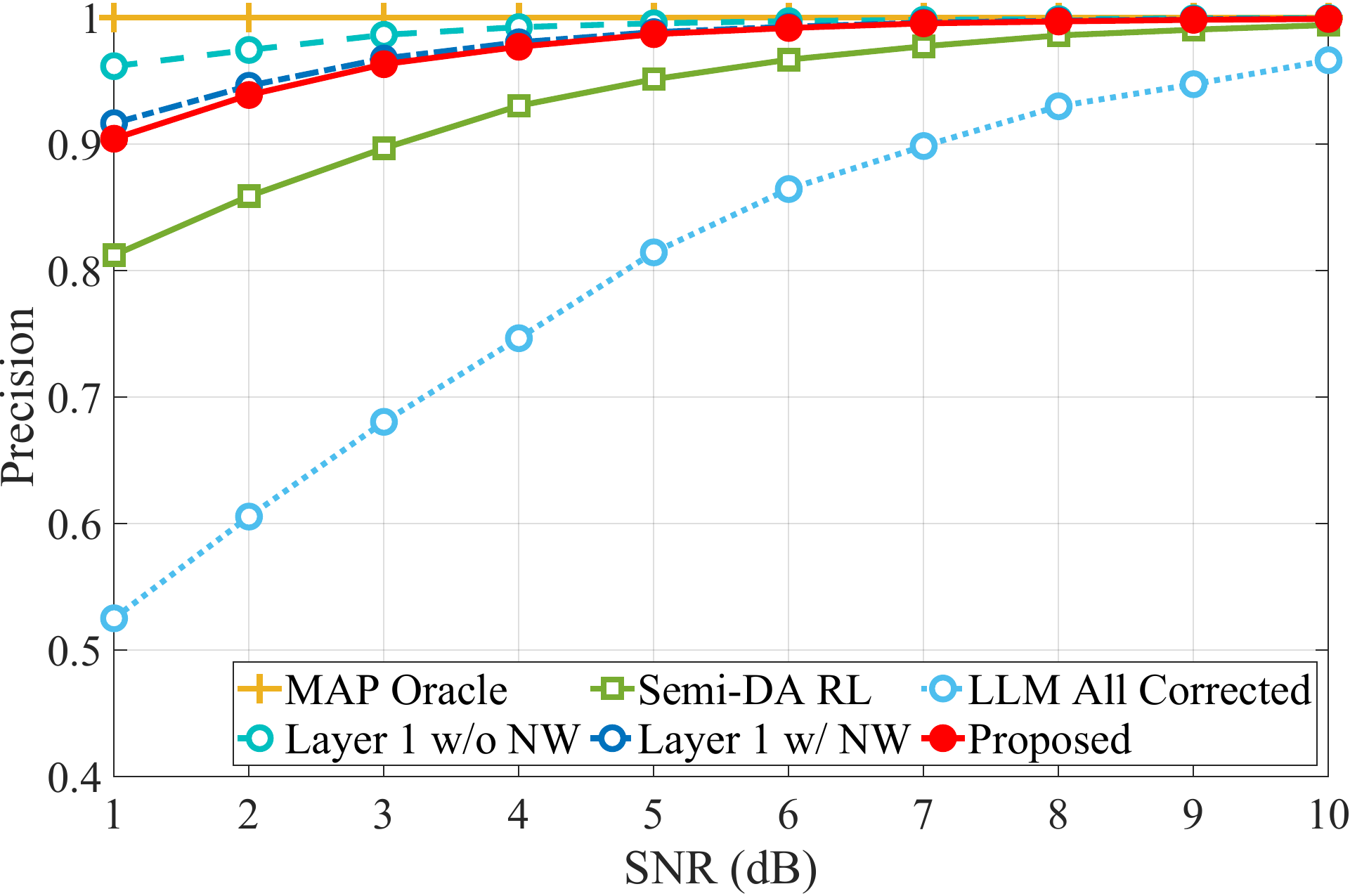}
        \caption{Precision}
        \label{fig:main_precision}
    \end{subfigure}
    
    \vspace{1mm}
    
    \begin{subfigure}{0.85\linewidth}
        \centering
        \includegraphics[width=\linewidth]{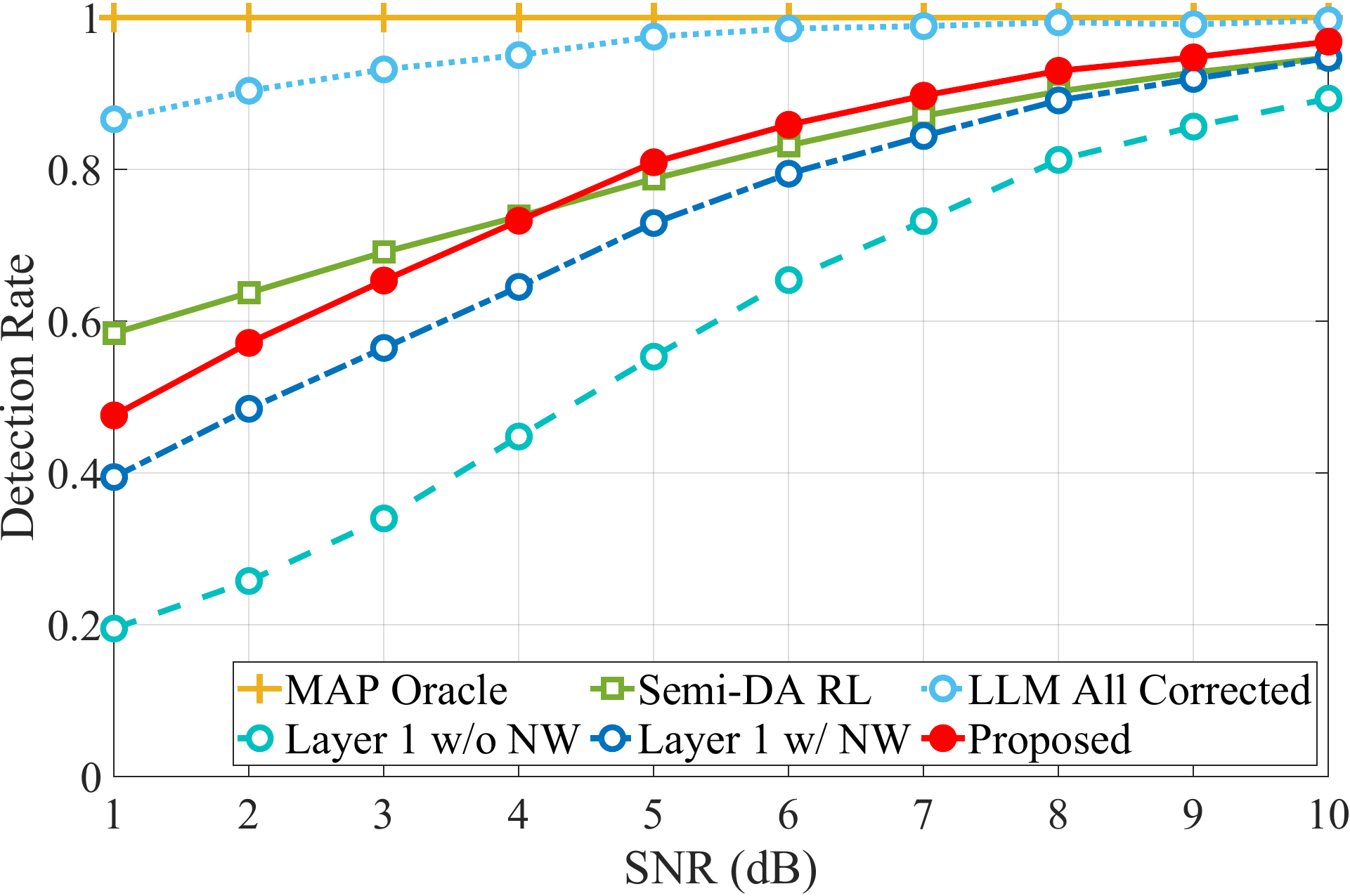}
        \caption{Detection Rate}
        \label{fig:main_detection_rate}
    \end{subfigure}
    \caption{Semantic pilot selection quality of the compared
methods versus SNR.}
    \label{fig:pilot_metrics}
\end{figure}

Fig.~\ref{fig:pilot_metrics} compares the pilot selection quality in terms of selection ratio, precision, and detection rate.
Fig.~\ref{fig:main_selection_ratio} shows the selection ratio, which represents the fraction of data symbols used as additional pilot. LLM All Corrected always achieves a selection ratio of 1, as it uses all data symbols without any filtering. MAP Oracle represents the fraction of correctly decoded symbols among all transmitted symbols, which serves as the upper bound for any selection method that relies on decoded symbols. The Proposed method achieves a higher selection ratio than both Layer~1 baselines, owing to the additional symbols contributed by Layer~2.
Fig.~\ref{fig:main_precision} shows the precision, defined as the fraction of selected pilot symbols that match the ground-truth transmitted symbols. MAP Oracle achieves a precision of 1 by definition. The Layer~1 baselines achieve the highest precision among practical methods, as they accept only positions where the decoded and corrected texts agree, which is inherently a strict criterion. However, as observed in Fig.~\ref{fig:main_selection_ratio}, this strictness results in a low selection ratio, limiting the number of reference symbols available for channel estimation. 
Fig.~\ref{fig:main_detection_rate} reports the detection rate, which measures the completeness of the selection rather than its reliability. MAP Oracle achieves a detection rate of 1 by definition. LLM All Corrected also exhibits a high detection rate, since using all corrected symbols. However, it does not reach 1 due to occasional LLM correction errors that displace correctly decoded symbols.
While Semi-DA RL achieves a higher detection rate than the proposed method at low SNR, it yields an inferior NMSE due to the lower precision shown in Fig.~\ref{fig:main_precision}.
The Layer~1 baselines yield the lowest detection rates due to their strict selection criteria. This approach ensures high precision but fails to utilize a substantial volume of symbols.
The proposed method achieves superior NMSE performance compared to the MAP Oracle despite utilizing fewer reference symbols. This behavior is consistent with Theorem~\ref{thm:ta_main}, as corrected symbols admitted through Layer~2 provide a greater NMSE reduction than the correctly decoded symbols.

\begin{figure}[t]
    \centering
    \includegraphics[width=0.85\linewidth]{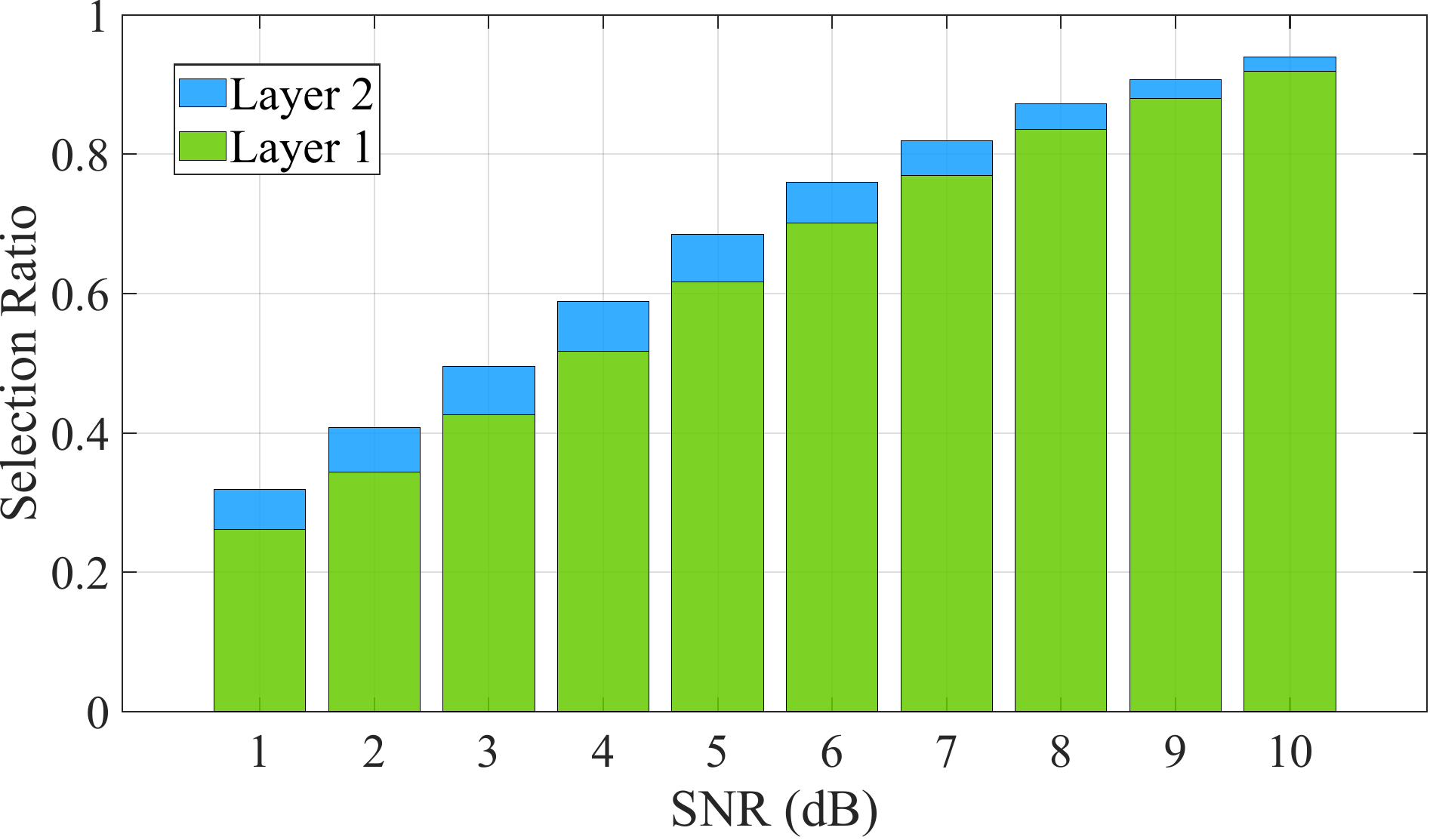}
    \caption{Composition of the semantic pilot versus SNR, showing Layer~1 and Layer~2 contributions.}
    \label{fig:layer_composition}
\end{figure}

\begin{figure}[t]
    \centering
    \includegraphics[width=0.85\linewidth]{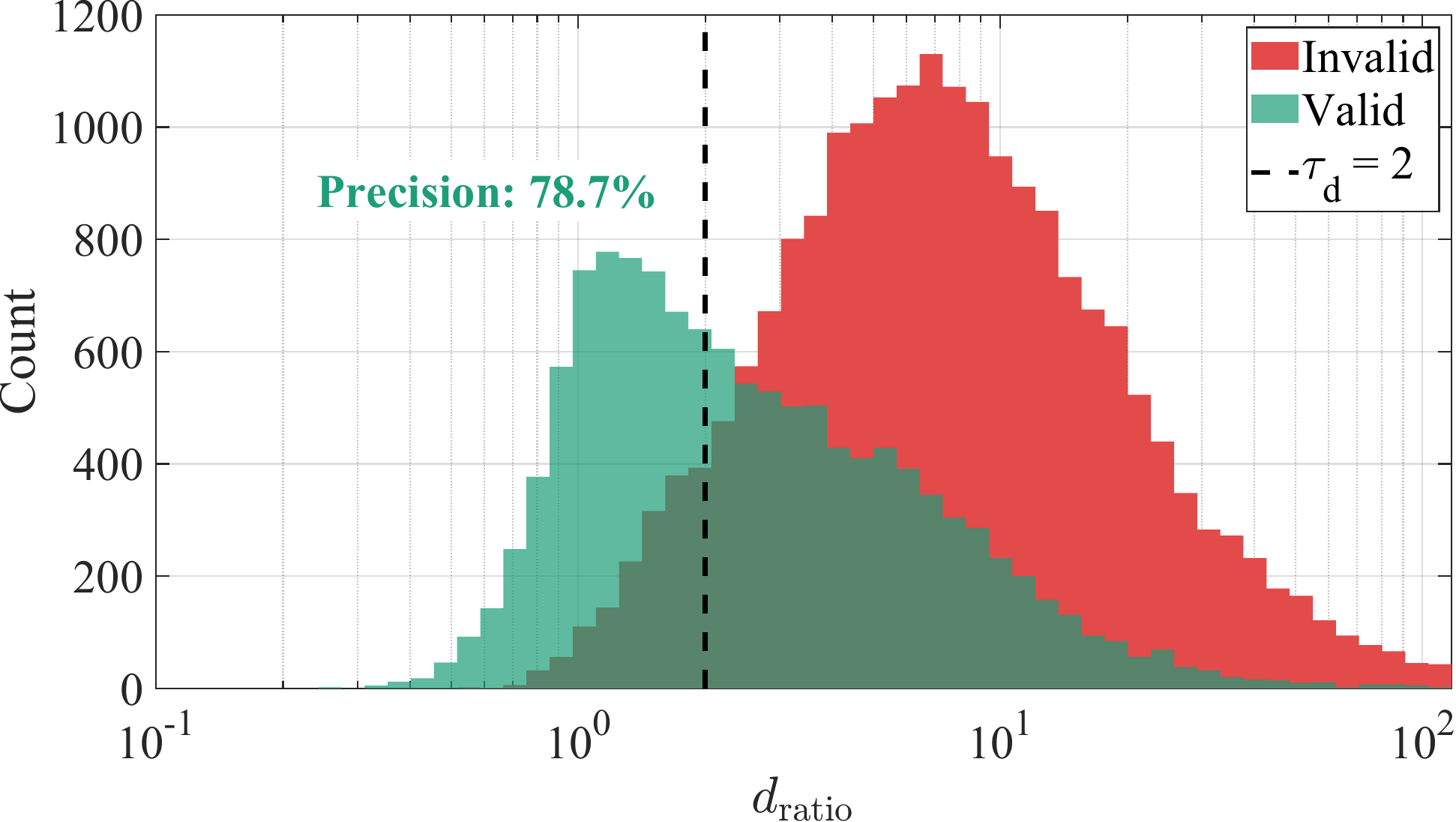}
    \caption{Distribution of $d_\mathrm{ratio}$ in Layer~2 cross-validation at 5~dB, separated into valid and invalid corrections.}
    \label{fig:dratio_hist}

    \vspace{3mm}

    \captionof{table}{Effect of $\tau_d$ on Layer~2 cross-validation
at 5~dB, evaluated on the hyperparameter tuning set.}
    \vspace{2mm}
    \label{tab:tau_sweep}
    \begin{tabular}{cccc}
    \toprule
    $\tau_d$ & Acceptance Rate (\%) & Precision (\%) & NMSE (dB) \\
    \midrule
    1 & 4.1  & 90.0 & $-19.72$ \\
    2 & 20.6 & 70.7 & $-21.31$ \\
    3 & 31.9 & 59.6 & $-20.85$ \\
    4 & 42.7 & 51.9 & $-19.31$ \\
    5 & 51.6 & 47.1 & $-17.67$ \\
    \bottomrule
    \end{tabular}
\end{figure}

\begin{figure*}[t]
    \centering
    \begin{subfigure}{0.195\linewidth}
        \centering
        \includegraphics[width=\linewidth]{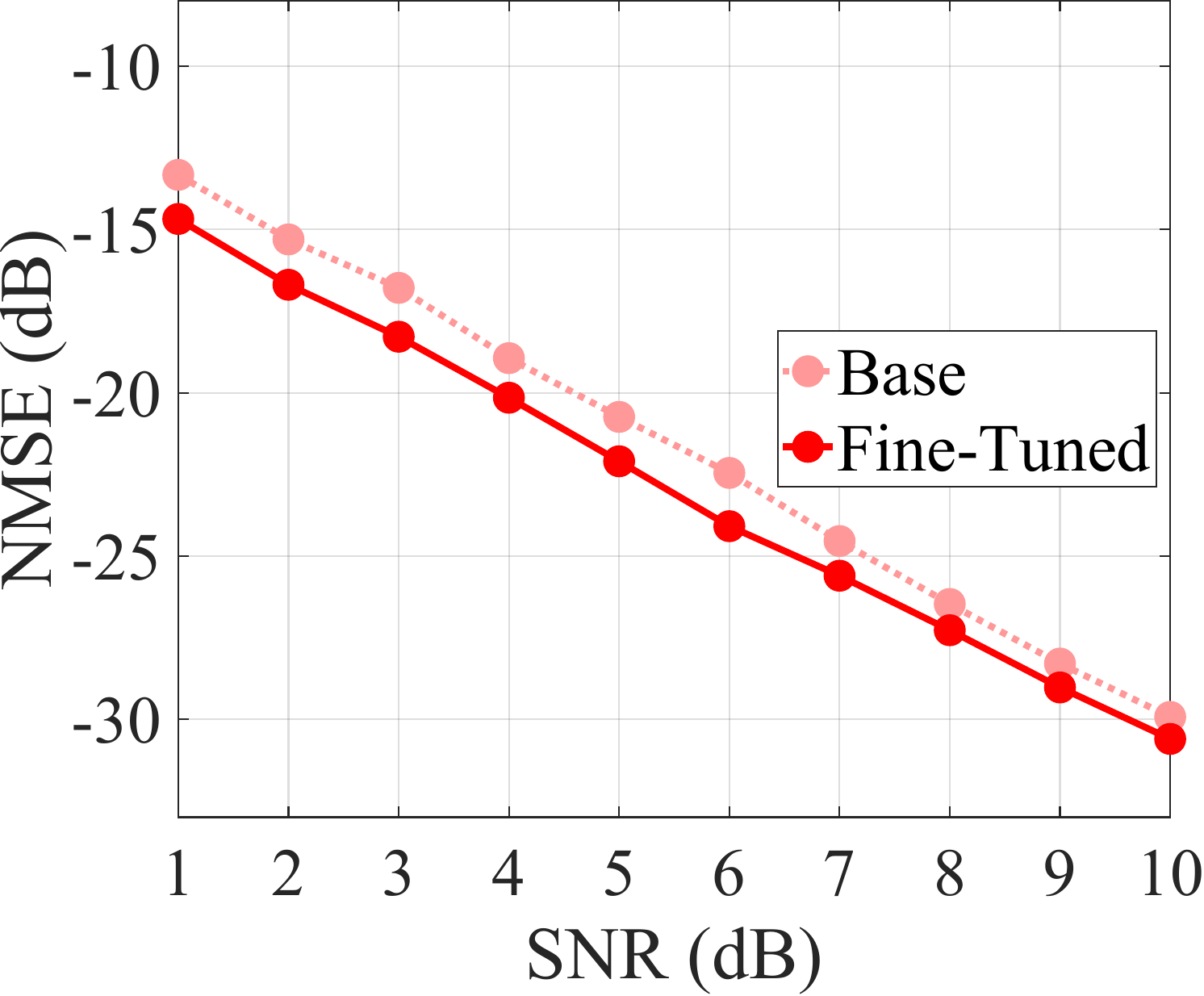}
        \caption{NMSE}
        \label{fig:nmse_abl}
    \end{subfigure}
    \hfill
    \begin{subfigure}{0.195\linewidth}
        \centering
        \includegraphics[width=\linewidth]{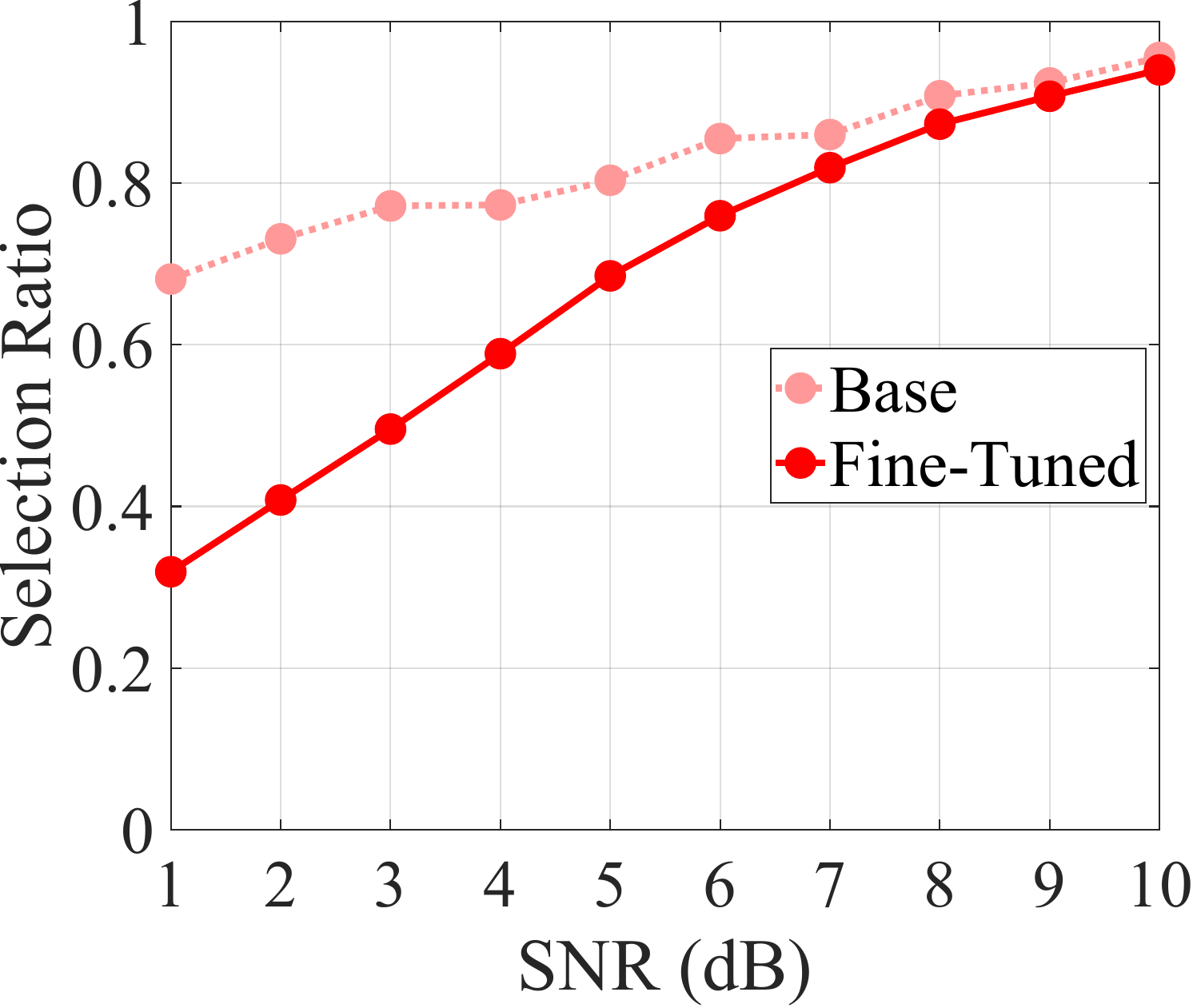}
        \caption{Selection Ratio}
        \label{fig:selction_ratio_abl}
    \end{subfigure}
    \hfill
    \begin{subfigure}{0.195\linewidth}
        \centering
        \includegraphics[width=\linewidth]{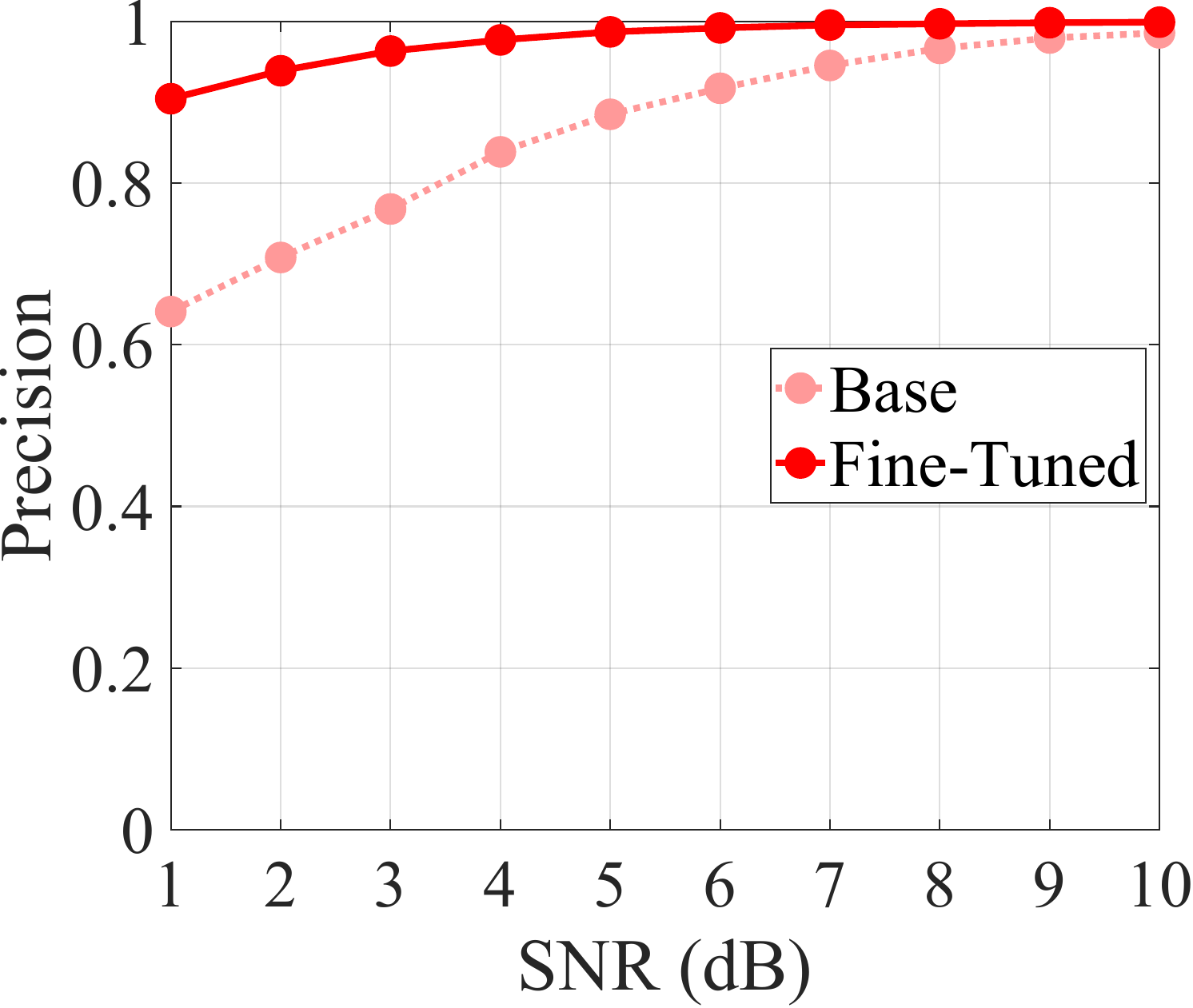}
        \caption{Precision}
        \label{fig:precision_abl}
    \end{subfigure}
    \hfill
    \begin{subfigure}{0.195\linewidth}
        \centering
        \includegraphics[width=\linewidth]{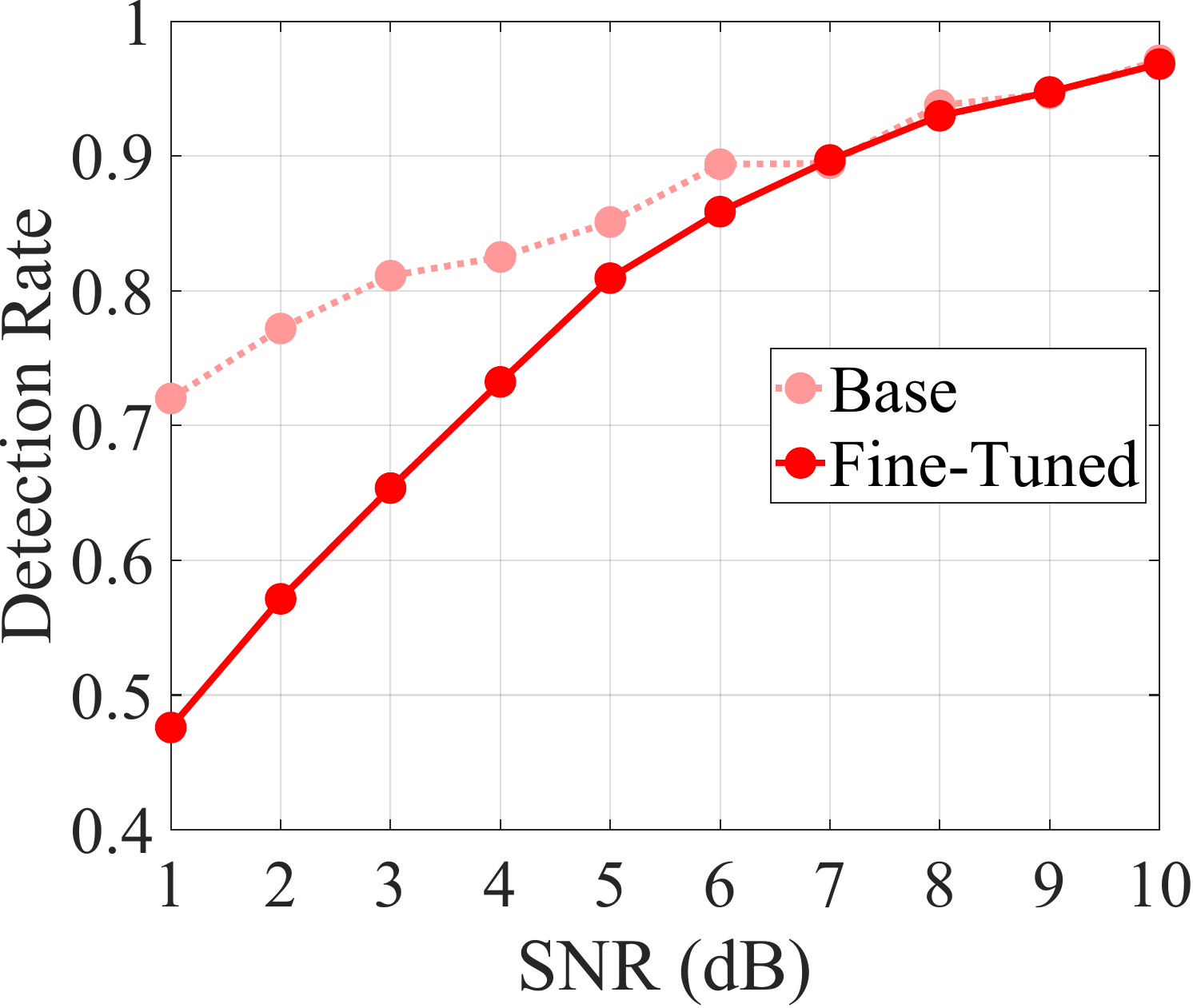}
        \caption{Detection Rate}
        \label{fig:detection_rate_abl}
    \end{subfigure}
    \hfill
    \begin{subfigure}{0.195\linewidth}
        \centering
        \includegraphics[width=\linewidth]{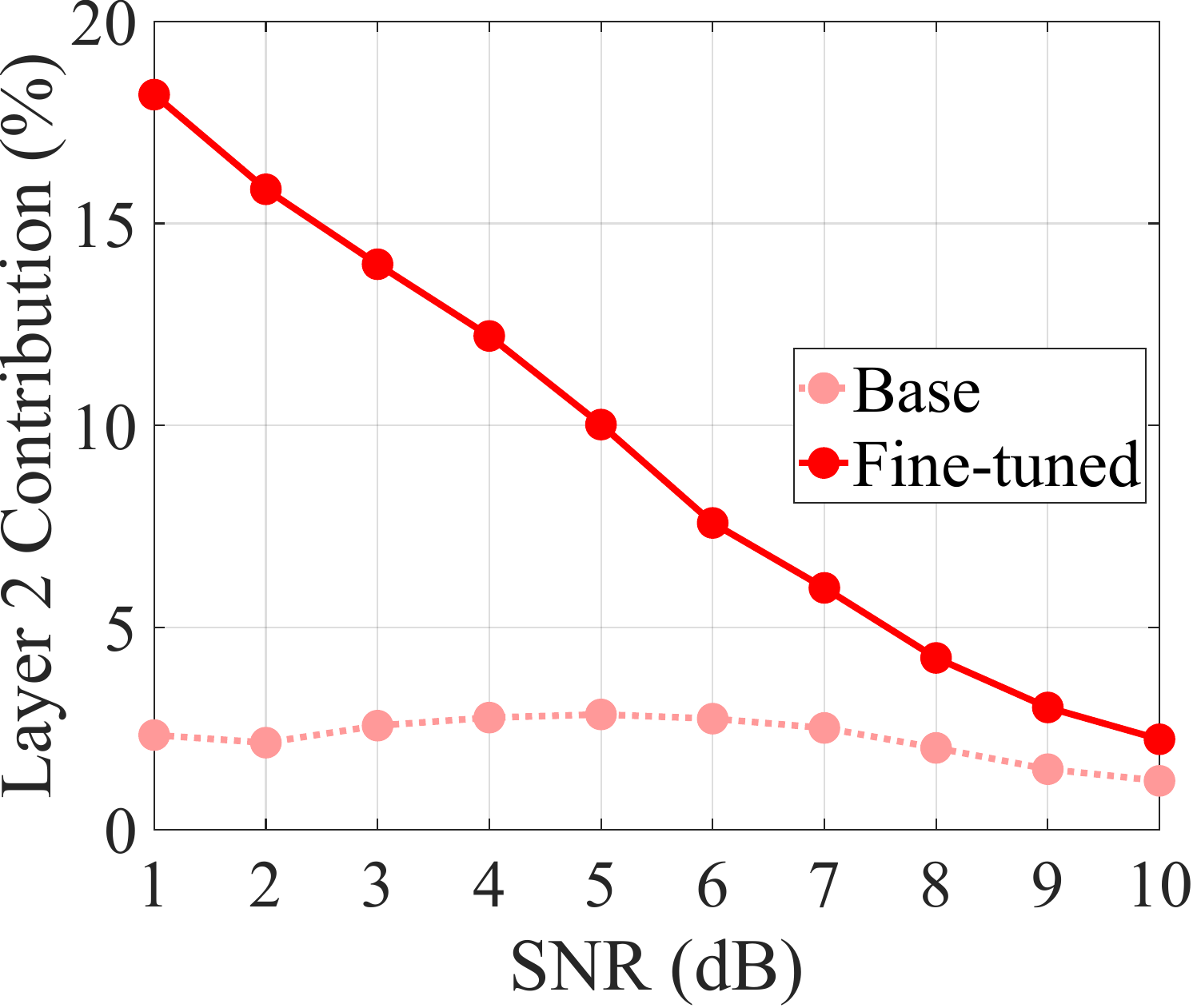}
        \caption{Layer 2 Contribution}
        \label{fig:l2_contribution_abl}
    \end{subfigure}
    \caption{Effect of fine-tuning on the proposed method.}
    \label{fig:ablation_ft}
\end{figure*}

Fig.~\ref{fig:layer_composition} illustrates the semantic pilot composition, with green and blue segments denoting the proportions of pilots selected by Layer 1 and Layer 2, respectively. In the low SNR regime, Layer 2 accounts for up to 18.2\% of the pilot symbols. This ratio diminishes at higher SNR levels as improved signal quality reduces the frequency of detection errors.

The threshold $\tau_d$ is selected using a hyperparameter tuning set. Table~\ref{tab:tau_sweep} reports the effect of $\tau_d$ on Layer~2 performance at 5~dB. As $\tau_d$ increases, more LLM corrections are accepted at the cost of lower precision. The NMSE is minimized at $\tau_d = 2$, which achieves a precision of 70.7\% with an acceptance rate of 20.6\%. We adopt this value across all SNRs, as it provides robust performance over the tested range.

Fig.~\ref{fig:dratio_hist} depicts the distribution of $d_\mathrm{ratio}$ for LLM-corrected symbols evaluated on the 5~dB test set. Symbols are classified as valid if the LLM-corrected version matches the ground-truth transmission and invalid otherwise. Although the two distributions partially overlap, they remain largely distinguishable, with valid corrections concentrated around low $d_\mathrm{ratio}$ values and invalid corrections spread across higher values. This separation confirms that $d_\mathrm{ratio}$ serves as an effective discriminator for filtering out incorrect LLM corrections. With the threshold $\tau_d = 2$, the cross-validation achieves a precision of 78.7\% among the accepted symbols at this SNR level.

\subsection{Effect of Fine-Tuning}
\label{subsec:ablation_ft}

To evaluate the impact of fine-tuning on the proposed method, we compare two configurations: Fine-Tuned, which uses the Qwen3-8B model adapted via QLoRA on the channel-corrupted training corpus, and Base, which uses the same Qwen3-8B model with identical quantization but without any task-specific adaptation. Fig.~\ref{fig:ablation_ft} presents the comparison across five metrics.
As shown in Fig.~\ref{fig:nmse_abl}, Fine-Tuned model consistently achieves a lower NMSE than Base model across all tested SNR levels. 
Fig.~\ref{fig:selction_ratio_abl} demonstrates that the Base model achieves a higher selection ratio than the Fine-Tuned model, particularly in the low SNR regime. However, Fig.~\ref{fig:precision_abl} reveals that this increased selection frequency is accompanied by significantly lower precision. Specifically, at 1~dB, the Base model yields a precision of 0.64, whereas the Fine-Tuned model maintains a precision of 0.90. These findings indicate that the Base model generates outputs that align with the decoded text without successfully resolving the underlying errors. This behavior results in a high volume of unreliable reference symbols that lack the precision necessary for effective channel refinement.
Fig.~\ref{fig:detection_rate_abl} illustrates that the Base model achieves a higher detection rate than the Fine-Tuned model in the low SNR regime. This result occurs as the Base model selects a larger total number of symbols as pilots. By selecting more symbols overall, the model naturally includes more references that are already correct, even if its actual ability to fix errors remains limited.
The most significant difference appears in Fig.~\ref{fig:l2_contribution_abl}. Fine-Tuned model achieves a Layer~2 contribution of 18.2\% at 1~dB, whereas Base model remains below 3\% across all SNR levels. This indicates that the Fine-Tuned model produces corrections that are sufficiently accurate to pass the $d_\mathrm{ratio}$ cross-validation, while the Base model's corrections are largely rejected as unreliable. Since Layer~2 is the primary source of the NMSE advantage, the low Layer~2 contribution of the Base model directly explains its inferior NMSE performance. These results confirm that fine-tuning is essential for the proposed framework to fully exploit the correction capability of the LLM.

\section{Conclusion}
\label{sec:conclusion}
In this paper, we proposed a novel semantic-aware data-aided channel estimation framework for MIMO systems. The proposed framework exploits the contextual information captured by the LLM to shift the paradigm of data-aided estimation from merely selecting reliable symbols to actively recovering unreliable symbols. To ensure robustness against potential LLM hallucinations, we developed a two-layer selection mechanism that combines semantic verification with physical-layer validation, supported by a NW alignment module to handle text length mismatches. Our theoretical analysis established a fundamental advantage of this approach, proving that recovering misdetected symbols provides a greater reduction in channel estimation error compared to using already correctly decoded symbols. Simulation results using a fine-tuned LLM confirmed that the proposed framework substantially improves channel estimation accuracy and detection performance, closely approaching the oracle bound. Overall, this work demonstrates that integrating payload semantics into physical-layer processing offers a highly effective pathway for enhancing the reliability of wireless communications.

\appendices

\section{Proof of Lemma~\ref{lem:ta_delta}}
\label{app:lemma_proof}

Let $\mathbf{y} = \mathbf{H}\mathbf{x} + \mathbf{n}$ denote the received signal of the additional pilot $\mathbf{x}$, where $\mathbf{n} \sim \mathcal{CN}(\mathbf{0},\sigma^2\mathbf{I}_{N_\mathrm{r}})$.
We assume that $\mathbf{n}$ is independent of the detection outcome that determines the reference type of $\mathbf{x}$. The expectation $\mathbb{E}_\mathbf{n}[\cdot]$ below is therefore taken over its nominal distribution.

Augmenting the pilot block with $(\mathbf{x},\mathbf{y})$, the LMMSE estimator becomes
\begin{align}
\hat{\mathbf{H}}_{\mathrm{new}} &= [\mathbf{Y}_\mathrm{p},\mathbf{y}][\mathbf{X}_\mathrm{p},\mathbf{x}]^H \Big([\mathbf{X}_\mathrm{p},\mathbf{x}][\mathbf{X}_\mathrm{p},\mathbf{x}]^H + \sigma^2\mathbf{I}_{N_\mathrm{t}}\Big)^{-1} \notag\\
&= \big(\mathbf{Y}_\mathrm{p}\mathbf{X}_\mathrm{p}^H + \mathbf{y}\mathbf{x}^H\big) \big(\mathbf{X}_\mathrm{p}\mathbf{X}_\mathrm{p}^H + \mathbf{x}\mathbf{x}^H + \sigma^2\mathbf{I}_{N_\mathrm{t}}\big)^{-1} \notag\\
&= \big(\mathbf{Y}_\mathrm{p}\mathbf{X}_\mathrm{p}^H + \mathbf{y}\mathbf{x}^H\big) \big((L_\mathrm{p}+\sigma^2)\mathbf{I}_{N_\mathrm{t}} + \mathbf{x}\mathbf{x}^H\big)^{-1}.
\label{eq:app_lmmse_expand}
\end{align}
The inverse in~\eqref{eq:app_lmmse_expand} is computed using $\|\mathbf{x}\|^2 = N_\mathrm{t}$ as
\begin{align}
&\big((L_\mathrm{p}+\sigma^2)\mathbf{I}_{N_\mathrm{t}} + \mathbf{x}\mathbf{x}^H\big)^{-1} \notag\\
&\quad= \frac{1}{L_\mathrm{p}+\sigma^2}\mathbf{I}_{N_\mathrm{t}} - \frac{\mathbf{x}\mathbf{x}^H}{(L_\mathrm{p}+\sigma^2)(L_\mathrm{p}+\sigma^2 + N_\mathrm{t})}.
\label{eq:app_inverse}
\end{align}
Substituting~\eqref{eq:app_inverse} into~\eqref{eq:app_lmmse_expand} and using $\hat{\mathbf{H}}_\mathrm{p} = \mathbf{Y}_\mathrm{p}\mathbf{X}_\mathrm{p}^H/(L_\mathrm{p}+\sigma^2)$ yields
\begin{equation}
\hat{\mathbf{H}}_{\mathrm{new}} = \hat{\mathbf{H}}_\mathrm{p} + \frac{(\mathbf{y} - \hat{\mathbf{H}}_\mathrm{p}\mathbf{x})\mathbf{x}^H}{L_\mathrm{p}+\sigma^2 + N_\mathrm{t}}.
\label{eq:app_update}
\end{equation}
Applying $\mathbf{y} = \mathbf{H}\mathbf{x} + \mathbf{n}$ and $\hat{\mathbf{H}}_\mathrm{p} = \mathbf{H} + \mathbf{E}$ into~\eqref{eq:app_update} gives
\begin{equation}
\hat{\mathbf{H}}_{\mathrm{new}} - \mathbf{H} = \mathbf{E}\Big(\mathbf{I}_{N_\mathrm{t}} - \frac{\mathbf{x}\mathbf{x}^H}{D}\Big) + \frac{\mathbf{n}\mathbf{x}^H}{D},
\label{eq:app_residual}
\end{equation}
where $D = L_\mathrm{p}+\sigma^2 + N_\mathrm{t}$. Taking the squared Frobenius norm of~\eqref{eq:app_residual} and the expectation over $\mathbf{n}$ yields
\begin{align}
&\mathbb{E}_{\mathbf{n}}\!\left[\|\hat{\mathbf{H}}_{\mathrm{new}} -\mathbf{H}\|_F^2 \,\big|\, \mathbf{E},\mathbf{x}\right] \notag\\
&\quad= \big\|\mathbf{E}(\mathbf{I}_{N_\mathrm{t}} - \tfrac{\mathbf{x}\mathbf{x}^H}{D})\big\|_F^2 + \frac{\sigma^2\,N_\mathrm{r}\,N_\mathrm{t}}{D^2}.
\label{eq:app_two_terms}
\end{align}
The first term in~\eqref{eq:app_two_terms} expands as
\begin{align}
&\big\|\mathbf{E}(\mathbf{I}_{N_\mathrm{t}} - \tfrac{\mathbf{x}\mathbf{x}^H}{D})\big\|_F^2 \notag\\
&\quad= \mathrm{tr}\!\left[\mathbf{E}\big(\mathbf{I}_{N_\mathrm{t}} - \tfrac{2\mathbf{x}\mathbf{x}^H}{D} + \tfrac{\mathbf{x}\mathbf{x}^H\mathbf{x}\mathbf{x}^H}{D^2}\big) \mathbf{E}^H\right] \\
&\quad= \mathrm{tr}\!\left[\mathbf{E}\big(\mathbf{I}_{N_\mathrm{t}} - \tfrac{2D - N_\mathrm{t}}{D^2}\mathbf{x}\mathbf{x}^H\big) \mathbf{E}^H\right] \\
&\quad= \|\mathbf{E}\|_F^2 - \frac{2D - N_\mathrm{t}}{D^2}\,R(\mathbf{x};\mathbf{E}).
\label{eq:app_first_term}
\end{align}
Combining~\eqref{eq:app_two_terms} and~\eqref{eq:app_first_term} gives
\begin{equation}
\mathbb{E}_{\mathbf{n}}\!\left[\|\hat{\mathbf{H}}_{\mathrm{new}} -\mathbf{H}\|_F^2 \,\big|\, \mathbf{E},\mathbf{x}\right] = \|\mathbf{E}\|_F^2 - a\,R(\mathbf{x};\mathbf{E}) + b,
\label{eq:app_cond_red}
\end{equation}
with $(a,b)$ as defined in~\eqref{eq:ta_ab}. Subtracting~\eqref{eq:app_cond_red} from $\|\mathbf{E}\|_F^2$ and taking the expectation over $\mathbf{E}$ yields
\begin{equation}
\begin{split}
\Delta(\mathbf{x}) &= \mathbb{E}\!\left[\|\mathbf{E}\|_F^2 - \mathbb{E}_{\mathbf{n}}\!\left[\|\hat{\mathbf{H}}_{\mathrm{new}} -\mathbf{H}\|_F^2 \,\big|\, \mathbf{E},\mathbf{x}\right]\right] \\
&= a\,\mathbb{E}\!\left[R(\mathbf{x};\mathbf{E})\,\big|\,\mathbf{x}\right] - b,
\end{split}
\label{eq:app_delta_final}
\end{equation}
which establishes~\eqref{eq:ta_delta}.\hfill$\blacksquare$

\section{Proof of Proposition~\ref{prop:ta_bias}}
\label{app:prop_proof}

By Lemma~\ref{lem:ta_delta}, both reference types share the same coefficients $(a,b)$ with $a>0$. Therefore, we only need to compare $\mathbb{E}[R(\mathbf{x};\mathbf{E})\mid\mathcal{C}]$ and $\mathbb{E}[R(\mathbf{x};\mathbf{E})\mid\mathcal{C}^{c}]$, where $\mathcal{C}$ denotes the event that the initial detector decodes the true symbol $\mathbf{x}$ correctly. We proceed in three steps.

\subsubsection*{Step 1: Independence of the pilot estimate and its error}
Under the orthogonal pilot $\mathbf{X}_\mathrm{p}\mathbf{X}_\mathrm{p}^H = L_\mathrm{p}\mathbf{I}_{N_\mathrm{t}}$, \eqref{eq:lmmse} gives
\begin{equation}
\hat{\mathbf{H}}_\mathrm{p} = \frac{L_\mathrm{p}\mathbf{H} + \mathbf{N}_\mathrm{p}\mathbf{X}_\mathrm{p}^H}{L_\mathrm{p}+\sigma^2},
\qquad
\mathbf{E} = \frac{-\sigma^2\mathbf{H} + \mathbf{N}_\mathrm{p}\mathbf{X}_\mathrm{p}^H}{L_\mathrm{p}+\sigma^2}.
\label{eq:app_hp_E}
\end{equation}
Let $\tilde{\mathbf{N}} = \mathbf{N}_\mathrm{p}\mathbf{X}_\mathrm{p}^H$, whose rows are zero-mean with covariance $\sigma^2 L_\mathrm{p}\mathbf{I}_{N_\mathrm{t}}$ and independent of $\mathbf{H}$. The cross-covariance between corresponding rows of $\hat{\mathbf{H}}_\mathrm{p}$ and $\mathbf{E}$ is
\begin{align}
&\mathbb{E}\!\big[(L_\mathrm{p}\mathbf{h}+\tilde{\mathbf{n}})(-\sigma^2\mathbf{h}+\tilde{\mathbf{n}})^H\big] \notag\\
&\quad= -\sigma^2 L_\mathrm{p}\mathbf{I}_{N_\mathrm{t}} + \sigma^2 L_\mathrm{p}\mathbf{I}_{N_\mathrm{t}} = \mathbf{0}.
\label{eq:app_crosscov}
\end{align}
Since $\hat{\mathbf{H}}_\mathrm{p}$ and $\mathbf{E}$ are jointly Gaussian, \eqref{eq:app_crosscov} implies that they are statistically independent. Consequently, the detection margin $\mathbf{d}_k = \hat{\mathbf{H}}_\mathrm{p}(\mathbf{x}-\mathbf{x}_k)$ is independent of $\mathbf{E}$.

Moreover, from \eqref{eq:app_hp_E} each row of $\mathbf{E}$ is $\mathcal{CN}(\mathbf{0},\alpha\mathbf{I}_{N_\mathrm{t}})$ with $\alpha = \sigma^2/(L_\mathrm{p}+\sigma^2)$, so for any deterministic $\mathbf{x}$ with $\|\mathbf{x}\|^2=N_\mathrm{t}$,
\begin{equation}
\mathbf{E}\mathbf{x}\sim\mathcal{CN}(\mathbf{0},\alpha N_\mathrm{t}\mathbf{I}_{N_\mathrm{r}}),
\qquad
\mathbf{E}\mathbf{x} = \sqrt{R}\,\mathbf{u},
\label{eq:app_isotropy}
\end{equation}
where $R = R(\mathbf{x};\mathbf{E}) = \|\mathbf{E}\mathbf{x}\|^2$ and $\mathbf{u}$ is uniform on the unit sphere, independent of $R$. In particular, the unconditional distribution of $R$ is identical for every true symbol. The difference between a correctly-decoded reference $\mathbf{x}_\mathrm{D}$ and a corrected reference $\mathbf{x}_\mathrm{C}$ therefore arises solely from conditioning on the detection outcome, i.e.,
\begin{equation}
\mathbb{E}[R(\mathbf{x}_\mathrm{D};\mathbf{E})] = \mathbb{E}[R\mid\mathcal{C}],
\quad
\mathbb{E}[R(\mathbf{x}_\mathrm{C};\mathbf{E})] = \mathbb{E}[R\mid\mathcal{C}^{c}].
\label{eq:app_cond_identity}
\end{equation}

\subsubsection*{Step 2: The correct-decision probability decreases in $R$}
Writing the effective noise as $\mathbf{w}=\mathbf{n}-\mathbf{E}\mathbf{x}$, the detection metric in~\eqref{eq:map} evaluated at the true symbol $\mathbf{x}$ and at a competing candidate $\mathbf{x}_k$ satisfies
\begin{equation}
\|\mathbf{y}-\hat{\mathbf{H}}_\mathrm{p}\mathbf{x}\|^2 = \|\mathbf{w}\|^2, \qquad
\|\mathbf{y}-\hat{\mathbf{H}}_\mathrm{p}\mathbf{x}_k\|^2 = \|\mathbf{w}+\mathbf{d}_k\|^2.
\label{eq:app_metrics}
\end{equation}
The event $\mathcal{C}$ occurs when the true symbol attains the smallest metric,
\begin{equation}
\|\mathbf{w}\|^2 < \|\mathbf{w}+\mathbf{d}_k\|^2, \qquad \forall k.
\label{eq:app_minmetric}
\end{equation}
Rearranging~\eqref{eq:app_minmetric} yields
\begin{equation}
2\,\mathrm{Re}\{\langle\mathbf{w},\mathbf{d}_k\rangle\} + \|\mathbf{d}_k\|^2 > 0,\qquad\forall k,
\label{eq:app_correct_cond}
\end{equation}
where $\mathrm{Re}\{\cdot\}$ denotes the real part of its argument.

In the operating SNR range, we adopt the standard nearest-neighbor approximation, under which the error event is dominated by the minimum-margin competitor. Denoting its margin by $\mathbf{d} = \arg\min_{k}\|\mathbf{d}_k\|$, the condition in~\eqref{eq:app_correct_cond} then reduces to the single pairwise test
\begin{equation}
2\,\mathrm{Re}\{\langle\mathbf{w},\mathbf{d}\rangle\} + \|\mathbf{d}\|^2 > 0.
\end{equation}
Substituting $\mathbf{w}=\mathbf{n}-\mathbf{E}\mathbf{x}$ separates the noise and the channel-error contributions,
\begin{equation}
2\,\mathrm{Re}\{\langle\mathbf{n},\mathbf{d}\rangle\} > 2\,\mathrm{Re}\{\langle\mathbf{E}\mathbf{x},\mathbf{d}\rangle\} - \|\mathbf{d}\|^2.
\label{eq:app_pairwise_split}
\end{equation}
Conditioned on $(\mathbf{E},\mathbf{d})$, the term $2\,\mathrm{Re}\{\langle\mathbf{n},\mathbf{d}\rangle\}$ is a real Gaussian random variable with zero mean and variance $2\sigma^2\|\mathbf{d}\|^2$, since $\mathbf{n}\sim\mathcal{CN}(\mathbf{0},\sigma^2\mathbf{I})$. Therefore, the probability that~\eqref{eq:app_pairwise_split} holds is
\begin{equation}
\Pr\{\mathcal{C}\mid\mathbf{E},\mathbf{d}\} = \Phi\!\left(\frac{\|\mathbf{d}\|^2 - 2\,\mathrm{Re}\{\langle\mathbf{E}\mathbf{x},\mathbf{d}\rangle\}}{\sqrt{2}\,\sigma\|\mathbf{d}\|}\right),
\label{eq:app_phi}
\end{equation}
where $\Phi(\cdot)$ is the standard normal cumulative distribution function, with derivative $\phi(\cdot) = \Phi'(\cdot)$ denoting the standard normal probability density function. Substituting $\mathbf{E}\mathbf{x} = \sqrt{R}\,\mathbf{u}$ into~\eqref{eq:app_phi} and letting $\gamma = \|\mathbf{d}\|/(\sqrt{2}\sigma) > 0$ and $t = (\sqrt{2}/\sigma)\,\mathrm{Re}\langle\mathbf{u}, \mathbf{d}/\|\mathbf{d}\|\rangle$ yields
\begin{equation}
\Pr\{\mathcal{C}\mid R,\mathbf{d},\mathbf{u}\} = \Phi\big(\gamma - \sqrt{R}\,t\big).
\label{eq:app_phi2}
\end{equation}
Since $\mathbf{u}$ is independent of $\mathbf{d}$ and isotropic, $t$ is symmetric about $0$. Define $p(r) = \Pr\{\mathcal{C}\mid R=r\} = \mathbb{E}_{\mathbf{d},t}[\Phi(\gamma-\sqrt{r}\,t)]$. Differentiating and pairing $t$ with $-t$,
\begin{align}
\frac{d p(r)}{d r}
&= -\frac{1}{2\sqrt{r}}\,\mathbb{E}_{\mathbf{d}}\,\mathbb{E}_{t>0}\!\Big[t\big(\phi(\gamma-\sqrt{r}\,t)-\phi(\gamma+\sqrt{r}\,t)\big)\Big].
\label{eq:app_deriv}
\end{align}
For $\gamma,r,t>0$, we have $(\gamma+\sqrt{r}\,t)^2-(\gamma-\sqrt{r}\,t)^2 = 4\gamma\sqrt{r}\,t>0$. This implies $\phi(\gamma+\sqrt{r}\,t)<\phi(\gamma-\sqrt{r}\,t)$, which makes the bracket in~\eqref{eq:app_deriv} strictly positive. Therefore, $p'(r)<0$, and the correct-decision probability is a strictly decreasing function of $R$.

\subsubsection*{Step 3: Selection bias via a correlation inequality}
Let $g(R) = p(R)$, which is strictly decreasing by Step~2. By the law of total probability, $\mathbb{E}[g(R)] = \mathbb{E}[\Pr\{\mathcal{C}\mid R\}] = \Pr\{\mathcal{C}\}$, which we denote by $P_\mathrm{c}\in(0,1)$. For two i.i.d. copies $R,R'$, monotonicity gives $(R-R')\big(g(R)-g(R')\big)\le 0$ almost surely, with strict inequality on a set of positive probability. Taking expectations,
\begin{equation}
\mathbb{E}[R\,g(R)] - \mathbb{E}[R]\,\mathbb{E}[g(R)] < 0.
\label{eq:app_chebyshev}
\end{equation}
Note that $g(R) = \Pr\{\mathcal{C}\mid R\} = \mathbb{E}[\mathds{1}_{\mathcal{C}}\mid R]$, where $\mathds{1}_{\mathcal{C}}$ is the indicator of the event $\mathcal{C}$. Hence $\mathbb{E}[R\,g(R)] = \mathbb{E}[R\,\mathds{1}_{\mathcal{C}}] = P_\mathrm{c}\,\mathbb{E}[R\mid\mathcal{C}]$, and~\eqref{eq:app_chebyshev} reduces to $\mathbb{E}[R\mid\mathcal{C}] < \mathbb{E}[R]$.
Applying the same argument to the strictly increasing function $1-g(R)$ gives $\mathbb{E}[R\mid\mathcal{C}^{c}] > \mathbb{E}[R]$. Combining these with the identity~\eqref{eq:app_cond_identity} of Step~1 yields
\begin{equation}
\mathbb{E}\big[R(\mathbf{x}_\mathrm{D};\mathbf{E})\big] < \mathbb{E}[R] < \mathbb{E}\big[R(\mathbf{x}_\mathrm{C};\mathbf{E})\big],
\end{equation}
which establishes~\eqref{eq:ta_bias}.\hfill$\blacksquare$

\bibliographystyle{IEEEtran}
\bibliography{strings,refs}




\end{document}